\documentclass[lettersize,journal]{IEEEtran}
\usepackage{amsmath,amsfonts}

\usepackage{algorithm}
\usepackage{array}
\usepackage[caption=false,font=normalsize,labelfont=sf,textfont=sf]{subfig}
\usepackage{textcomp}
\usepackage{stfloats}
\usepackage{url}
\usepackage{verbatim}
\usepackage{graphicx}
\usepackage{cite}
\usepackage{booktabs}
\usepackage{amsmath}
\usepackage{amssymb}
\usepackage{amsthm}
\newtheorem{definition}{Definition}
\newtheorem{proposition}{Proposition}
\newtheorem{theorem}{Theorem}

\usepackage{booktabs}
\usepackage{tabularx}
\usepackage{graphicx}
\usepackage{algpseudocode}
\usepackage{changepage}
\begin{document}

\title{Optimal Sampling and Reconstruction of Graph Signals in the Fractional Fourier Domain}

\author{Xiaopeng~Cheng, Zhichao~Zhang,~\IEEEmembership{Member,~IEEE}, and Yangfan~He
	\thanks{This work was supported in part by the Open Foundation of Hubei Key Laboratory of Applied Mathematics (Hubei University) under Grant HBAM202404; in part by the Foundation of Key Laboratory of System Control and Information Processing, Ministry of Education under Grant Scip20240121; in part by the Startup Foundation for Introducing Talent of Nanjing Institute of Technology under Grant YKJ202214; and in part by the Open Foundation of Jiangsu Province Engineering Research Center of IntelliSense Technology and System under Grant ITS202502. \emph{(Corresponding author: Zhichao~Zhang.)}}
	\thanks{Xiaopeng~Cheng is with the School of Mathematics and Statistics, Nanjing University of Information Science and Technology, Nanjing 210044, China (e-mail: cxp\_class666@163.com).}
	\thanks{Zhichao~Zhang is with the School of Mathematics and Statistics, Nanjing University of Information Science and Technology, Nanjing 210044, China, with the Hubei Key Laboratory of Applied Mathematics, Hubei University, Wuhan 430062, China, and also with the Key Laboratory of System Control and Information Processing, Ministry of Education, Shanghai Jiao Tong University, Shanghai 200240, China (e-mail: zzc910731@163.com).}
	\thanks{Yangfan~He is with the School of Communication and Artificial Intelligence, School of Integrated Circuits, Nanjing Institute of Technology, Nanjing 211167, China, and also with the Jiangsu Province Engineering Research Center of IntelliSense Technology and System, Nanjing 211167, China (e-mail: Yangfan.He@njit.edu.cn).}}

\markboth{}%
{Shell \MakeLowercase{\textit{et al.}}: A Sample Article Using IEEEtran.cls for IEEE Journals}

\IEEEpubid{}

\maketitle

\begin{abstract}
Graph signal sampling and reconstruction are commonly formulated in the graph Fourier transform (GFT) domain. However, the reconstruction performance may be limited when practical graph signals are not sufficiently concentrated in the GFT spectrum. To address this issue, this paper proposes a graph signal sampling and reconstruction framework based on the graph fractional Fourier transform (GFRFT) domain. The fractional order is introduced as an adjustable spectral domain parameter, and the optimal order is selected to provide a more suitable representation domain for a given graph signal and sampling model. Under a unified sampling reconstruction formulation, subspace, smoothness, and stochastic priors are incorporated, and both unconstrained and predefined reconstruction mechanisms are considered, leading to several fractional domain sampling and reconstruction methods. Furthermore, the theoretical analysis shows that the optimal GFRFT domain can provide a more suitable low-dimensional spectral representation by improving energy concentration and reducing projection residual. The effects of residual leakage and noise amplification are further considered to explain how this representation advantage is translated into reconstruction error reduction. Experimental results show that, GFRFT domain sampling and reconstruction generally achieve better recovery performance than GFT domain methods.
\end{abstract}


\section{Introduction}
\IEEEPARstart{G}{raph} signal processing (GSP) provides a signal analysis framework for data defined on graph structures~\cite{1,2,3,4,5,6}. Its main idea is to extend the representation, transform, filtering, and analysis methods in traditional signal processing (SP) to graph domains. Different from SP, a graph signal is defined on the nodes of a graph, and the graph structure is used to describe the relationships among data~\cite{4,5,6}. In recent years, with the aid of tools such as the graph Fourier transform (GFT)~\cite{7,8,9} and graph spectral analysis~\cite{10}, GSP has provided a unified processing method for complex network data and has gradually formed a relatively systematic theoretical framework covering representation, transform, filtering~\cite{10,11}, neural networks~\cite{12,13}, sampling and reconstruction~\cite{11,14,15,16,17,18,19}.

Among the various research topics in GSP, graph signal sampling and reconstruction is one of the core problems~\cite{15,16,17,18}. Its objective is to effectively recover the complete graph signal from only partial node observations by incorporating graph structural information and signal priors. This problem is of great significance in sensor data completion, traffic state estimation, network monitoring, and distributed information processing~\cite{22,23,24}. Existing related methods are mostly developed in the vertex domain or the GFT domain~\cite{20,21}, and rely on assumptions such as bandlimitedness and prior information to achieve signal recovery under certain conditions. Although these methods have laid an important foundation for graph signal sampling theory, the GFT is fixed by the graph structure, and its spectral basis lacks adjustability. For practical graph signals, their energy may not always be sufficiently concentrated in the GFT domain. When the spectral energy distribution of a signal is relatively dispersed, large spectral truncation errors or reconstruction errors may occur under limited sampling conditions, thereby affecting the recovery performance.

As a generalization of the GFT, the graph fractional Fourier transform (GFRFT) introduces a fractional order parameter and provides a class of continuously adjustable representation domains for graph signals~\cite{25,26,27,28,29}. When the fractional order takes a specific value, the GFRFT can reduce to the GFT. For other fractional orders, the same graph signal may exhibit different spectral coefficient distributions and low-dimensional energy concentration characteristics. Therefore, the fractional order can be regarded as an adjustable spectral domain parameter~\cite{291,292}. By searching for the optimal fractional order, a more suitable spectral representation domain can be selected for a given graph signal, sampling condition, and reconstruction model, which provides a new idea for improving graph signal reconstruction performance under limited sampling conditions.

In recent years, graph signal sampling and reconstruction have been extensively studied. Early works mainly focused on bandlimited graph signals, aiming to extend the Shannon-Nyquist sampling principle to graph signal processing and to investigate sampling set selection, recovery conditions, and algorithm design~\cite{30}. However, due to the irregularity of graph structures, the notions of frequency and bandlimitedness in the graph domain are fundamentally different from those in classical signal processing, which brings new challenges to both theoretical modeling and practical implementation. To overcome the limitation of the strict bandlimited assumption, generalized sampling theory (GST) has been introduced into graph signal processing to handle more general non-bandlimited graph signals~\cite{31,32,33,34,35,36}. Representative studies established GST frameworks in the graph frequency domain by extending the GST theory for shift-invariant subspaces to the graph setting~\cite{37}. In these frameworks, subspace and smoothness priors were considered, enabling GST to go beyond strictly bandlimited graph signals. Subsequent studies further extended this line of research in different directions. On the one hand, stochastic priors were incorporated into GST, where correction filters were designed by minimizing the mean-square error~\cite{38}. On the other hand, sampling operator design in the vertex domain was also investigated~\cite{39}, with vertex-wise sampling being extended to flexible sampling to improve the robustness of sampling and recovery under arbitrary priors. These studies show that GST has evolved from deterministic bandlimited models to more general formulations involving subspace, smoothness, and stochastic priors\cite{34,39,40}. Nevertheless, most existing methods are still formulated in the GFT domain or the vertex domain, leaving room for further investigation of graph signal sampling in more flexible representation domains.

On this basis, the GFRFT has also been introduced into graph signal sampling and reconstruction. Early studies on GFRFT domain sampling mainly focused on fractional bandlimited graph signals. Wang and Li studied sampling and recovery under the GFRFT and showed that \(a\)-bandlimited graph signals in the GFRFT domain can be perfectly recovered~\cite{25}. They also designed fractional sampling operators and demonstrated that fractional sampling may achieve better classification performance at an optimal fractional order than GFT sampling. Later, GST ideas were further extended to the GFRFT domain. Since most real graph signals are not strictly bandlimited in the GFRFT domain, Wei and Yan proposed a GST framework based on prior information in the GFRFT domain, where sampling and reconstruction are no longer restricted by the fractional domain bandwidth~\cite{40}. Their framework considers both subspace prior and smoothness prior, enabling the recovery of graph signals with or without bandlimitedness in the GFRFT domain.

However, existing studies have mainly focused on the construction and extension of GFRFT domain sampling and reconstruction models~\cite{25,27,40}, while systematic discussions on how to select the fractional order and whether the optimal fractional domain can outperform the GFT domain under different sampling and reconstruction conditions remain insufficient. In particular, the influence of fractional order variations on reconstruction errors, the performance of different priors and reconstruction mechanisms in the fractional domain, and the applicability of related methods to real graph data still require further investigation. Based on this consideration, this paper regards the fractional order as an adjustable spectral domain parameter and searches for the optimal order to select a more suitable spectral representation space for different graph signals. On this basis, a unified sampling-correction-reconstruction framework is further constructed, and the reconstruction advantages of the GFRFT domain over the GFT domain are analyzed from both theoretical and experimental perspectives.

The main contributions of this paper are summarized as follows.

\begin{itemize}
	\item This paper introduces an optimal GFRFT domain selection strategy by treating the fractional order as an adjustable spectral domain parameter, thereby selecting a more suitable representation space for graph signal reconstruction under limited sampling conditions.
	
	\item This paper develops a unified GFRFT domain sampling-correction-reconstruction framework by incorporating subspace, smoothness, and stochastic priors under both unconstrained and predefined reconstruction settings.
	
	\item This paper provides a theoretical interpretation from the perspectives of fractional spectral subspace, energy concentration, projection residual, residual leakage, and noise amplification, explaining the reconstruction advantage of the fractional domain representation.
\end{itemize}

The remainder of this paper is organized as follows. Section~II introduces the preliminaries of graph signals, the GFT, the GFRFT, and graph signal priors. Section~III develops the GFRFT domain sampling reconstruction framework under different priors. Section~IV analyzes the mechanism of GFRFT domain reconstruction from the perspectives of fractional spectral subspace, energy concentration, projection residual, residual leakage, and noise amplification. Section~V presents the experimental results on simulated graph signals and real traffic data. Section~VI concludes this paper.
\section{Preliminaries}
\subsection{GFT}

Let $\mathcal{G}=(\mathcal{V},\mathcal{E},\mathbf{A})$ be an undirected weighted graph with $N$ nodes, where $\mathbf{A}\in\mathbb{R}^{N\times N}$ denotes the adjacency matrix. A graph signal is defined as
\begin{equation*}
	\mathbf{x}=[x_1,x_2,\ldots,x_N]^\mathrm{T}\in\mathbb{R}^{N},
\end{equation*}
where $x_i$ is the signal value on the $i$-th node.

The combinatorial graph Laplacian is defined as
\begin{equation*}
	\mathbf{L}=\mathbf{D}-\mathbf{A},
\end{equation*}
where $\mathbf{D}$ is the degree matrix. Since $\mathbf{L}$ is symmetric for an undirected graph, it can be decomposed as
\begin{equation*}
	\mathbf{L}=\mathbf{U}\mathbf{\Lambda}\mathbf{U}^{\mathrm{T}},
\end{equation*}
where $\mathbf{U}$ contains the orthonormal eigenvectors and $\mathbf{\Lambda}$ is the diagonal eigenvalue matrix.

The GFT and its inverse are defined as
\begin{equation*}
	\widehat{\mathbf{x}}=\mathbf{U}^{\mathrm{T}}\mathbf{x},
\end{equation*}
\begin{equation*}
	\mathbf{x}=\mathbf{U}\widehat{\mathbf{x}},
\end{equation*}
where $\hat{\mathbf{x}}$ denotes the graph spectral coefficients. In this representation, the eigenvectors of $\mathbf{L}$ form the graph Fourier basis, and the eigenvalues correspond to graph frequencies.

\subsection{GFRFT}
Let the GFT matrix be
\begin{equation*}
	\mathbf{F}=\mathbf{U}^{\mathrm{T}}.
\end{equation*}
Assume that $\mathbf{F}$ is diagonalizable as
\begin{equation*}
	\mathbf{F}=\mathbf{V}\mathbf{\Lambda}_{F}\mathbf{V}^{-1},
\end{equation*}
where $\mathbf{\Lambda}_{F}$ contains the eigenvalues of $\mathbf{F}$. The GFRFT matrix with fractional order $a$ is defined as
\begin{equation}
	\mathbf{F}^{a}=\mathbf{V}\mathbf{\Lambda}_{F}^{a}\mathbf{V}^{-1},
\end{equation}
where
\begin{equation*}
	\mathbf{\Lambda}_{F}^{a}
	=
	\mathrm{diag}\left(\lambda_{F,1}^{a},\lambda_{F,2}^{a},\ldots,\lambda_{F,N}^{a}\right).
\end{equation*}

For a graph signal $\mathbf{x}$, the GFRFT and its inverse are given by
\begin{equation}
	\label{eq:xa_estimate}
	\widehat{\mathbf{x}}_{a}=\mathbf{F}^{a}\mathbf{x},
\end{equation}
\begin{equation}
	\mathbf{x}=\mathbf{F}^{-a}\widehat{\mathbf{x}}_{a}.
\end{equation}
When $a=1$, the GFRFT reduces to the conventional GFT, while different values of $a$ define different fractional graph spectral domains. Therefore, the fractional order provides an adjustable parameter for selecting a suitable transform domain in graph signal reconstruction.

\subsection{Prior Assumptions for Graph Signal Reconstruction}
\subsubsection{Subspace Prior}

The subspace prior assumes that a graph signal can be represented by a known low dimensional generation space~\cite{35,39,40}. For a graph signal $\mathbf{x}\in\mathbb{R}^{N}$, this prior can be written as
\begin{equation}
	\mathbf{x}=\mathbf{M}\mathbf{d},
\end{equation}
where $\mathbf{M}\in\mathbb{R}^{N\times K}$ is a known generator matrix, $\mathbf{d}\in\mathbb{R}^{K}$ is the expansion coefficient vector, and $K\leq N$. This assumption indicates that although the graph signal lies in an $N$-dimensional space, its effective degrees of freedom can be described by a smaller number of coefficients. Therefore, the reconstruction of the graph signal is converted into the estimation of a low dimensional coefficient vector. The key requirement of this prior is that the generation space should be specified in advance, which makes it a relatively strong structural assumption.

\subsubsection{Smoothness Prior}

The smoothness prior assumes that the graph signal has bounded energy under a certain variation measure~\cite{35,39,40}. Unlike the subspace prior, it does not require the signal to belong to a known low dimensional subspace. Instead, it constrains the overall variation of the signal, which can be expressed as
\begin{equation}
	\|\mathbf{T}\mathbf{x}\|_2^2 \leq \rho^2,
\end{equation}
where $\mathbf{N}$ is an operator used to measure the signal variation, and $\rho$ is a positive constant controlling the variation range. This prior focuses on the boundedness of signal variation rather than the exact form of the signal generation space. Therefore, it provides a weaker structural constraint when the underlying signal subspace is difficult to determine. By suppressing unreasonable variations, the smoothness prior encourages the reconstructed signal to follow a prescribed variation pattern.

\subsubsection{Stochastic Prior}

The stochastic prior models the graph signal as a random vector or random process characterized by statistical quantities~\cite{39}. A common assumption is that the graph signal has a given mean and covariance, for example,
\begin{equation}
	\mathbf{x}\sim \mathcal{N}(\mathbf{0},\boldsymbol{\Gamma}_x),
\end{equation}
where $\boldsymbol{\Gamma}_x$ is the covariance matrix of the graph signal. It describes the overall properties of a class of graph signals through statistical characteristics. These characteristics may include the mean structure, correlation structure, covariance structure, or spectral energy distribution. The essential assumption of the stochastic prior is that graph signals are not arbitrary, but exhibit stable or describable behavior in a statistical sense. Under this prior, reconstruction usually relies on second order statistical information to characterize the correlations among different nodes or spectral components.

\section{Proposed GFRFT Domain Sampling and Reconstruction Framework}
\label{sec:proposed_framework}

This section presents a graph signal sampling and reconstruction framework in the GFRFT domain. Conventional methods are usually formulated in the GFT domain, where the spectral basis is fixed by the graph structure. When practical graph signals are not sufficiently concentrated in the GFT spectrum, large truncation and reconstruction errors may occur under limited sampling conditions. To improve the flexibility of spectral representation, the proposed framework introduces the fractional order into the sampling and reconstruction process, enabling the search for a transform domain that better matches the structure of the current graph signal.

\subsection{Baseline Sampling Reconstruction Model}
\label{subsec:baseline_sampling_reconstruction}

We first introduce the baseline graph signal sampling and reconstruction model in the GFT domain~\cite{35,39}. Let \(\mathbf{x}\in\mathbb{C}^{N}\) be the original graph signal and \(\mathbf{F}\) be the GFT matrix. In the spectral sampling process, the graph signal is first transformed into the GFT domain, then weighted by a sampling filter \(\mathbf{S}(\mathbf{G})\), and finally reduced to \(K\) samples by the downsampling operator \(\mathbf{D}_s\in\mathbb{C}^{K\times N}\). The baseline sampling operator is written as
\begin{equation}
	\mathbf{S}_1^{\ast}
	=
	\mathbf{D}_s
	\mathbf{S}(\mathbf{G})
	\mathbf{F}.
	\label{eq:baseline_sampling_operator}
\end{equation}
The sampled signal is then given by
\begin{equation}
	\mathbf{y}
	=
	\mathbf{S}_1^{\ast}
	\mathbf{x}
	+
	\boldsymbol{\eta},
	\label{eq:baseline_sampled_signal}
\end{equation}
where \(\boldsymbol{\eta}\) denotes measurement noise or sampling perturbation. To recover the original graph signal from the low-dimensional observation, a correction operator \(\mathbf{H}_1(\mathbf{G})\) and a reconstruction operator \(\mathbf{W}_1\) are introduced. The recovered signal is expressed as
\begin{equation}
	\widetilde{\mathbf{x}}_1
	=
	\mathbf{W}_1
	\mathbf{H}_1(\mathbf{G})
	\mathbf{S}^{\ast}
	\mathbf{x}
	+
	\boldsymbol{\eta}.
	\label{eq:baseline_sampling_reconstruction_model}
\end{equation}
This model describes the basic sampling-correction-reconstruction chain, where \(\mathbf{S}_1^{\ast}\) determines how the graph signal is sampled, \(\mathbf{H}_1(\mathbf{G})\) compensates for non-ideal sampling effects, and \(\mathbf{W}_1\) maps the corrected samples back to the vertex domain.

\subsection{Proposed GFRFT Domain Sampling Reconstruction Model}
\label{subsec:gfrft_sampling_reconstruction}

Based on the baseline model, we extend the sampling and reconstruction process from the GFT domain to the GFRFT domain. For a given fractional order \(a\), the graph signal is represented by \(\mathbf{F}^{a}\), and the fractional domain sampling filter is defined as
\begin{equation*}
	\mathbf{S}_a(\mathbf{G})
	=
	\operatorname{diag}
	\left(
	S_a(g_0),
	S_a(g_1),
	\ldots,
	S_a(g_{N-1})
	\right),
	\label{eq:gfrft_sampling_filter}
\end{equation*}
where \(g_i\) denotes the \(i\)-th fractional spectral component. With the downsampling operator \(\mathbf{D}_s\in\mathbb{C}^{K\times N}\), the sampling operator and the sampled signal are given by
\begin{equation}
	\mathbf{S}_a^{\ast}
	=
	\mathbf{D}_s
	\mathbf{S}_a(\mathbf{G})
	\mathbf{F}^{a},
	\label{eq:gfrft_sampling_operator}
\end{equation}
and
\begin{equation}
	\mathbf{y}_a
	=
	\mathbf{S}_a^{\ast}
	\mathbf{x}
	+
	\boldsymbol{\eta},
	\label{eq:gfrft_sampled_signal}
\end{equation}
where \(\boldsymbol{\eta}\) denotes measurement noise or sampling perturbation.

To compensate for spectral folding, noise, and model mismatch, a correction operator \(\mathbf{H}_a(\mathbf{G})\) is applied in the reduced fractional spectral domain. The overall reconstruction operator is defined as
\begin{equation}
	\mathbf{W}_a
	=
	\mathbf{F}^{-a}
	\mathbf{W}_a(\mathbf{G})
	\mathbf{D}_s^{\mathrm{H}}.
	\label{eq:overall_reconstruction_operator}
\end{equation}
Thus, the GFRFT domain sampling reconstruction model can be written as
\begin{equation}
	\widetilde{\mathbf{x}}_a
	=
	\mathbf{W}_a
	\mathbf{H}_a(\mathbf{G})
	\left(
	\mathbf{S}_a^{\ast}
	\mathbf{x}
	+
	\boldsymbol{\eta}
	\right).
	\label{eq:general_gfrft_sampling_model}
\end{equation}

Compared with the baseline model in \eqref{eq:baseline_sampling_reconstruction_model}, the proposed model introduces the fractional order \(a\) into the sampling, correction, and reconstruction operators. Therefore, the framework can search for a more suitable spectral representation domain. When \(a=1\), and the proposed model degenerates into the GFT domain sampling reconstruction model. For clarity, Table~\ref{tab:notation_summary} summarizes the main notations and abbreviations used in the proposed GFRFT domain sampling reconstruction framework and theoretical analysis. In particular, SS, SM, and ST denote the subspace, smoothness, and stochastic priors, respectively, while UNC and PD denote the unconstrained and predefined recovery modes, respectively.

\begin{table}[!t]
	\centering
	\caption{Summary of Main Notations and Abbreviations}
	\label{tab:notation_summary}
	\scriptsize
	\renewcommand{\arraystretch}{1.14}
	\setlength{\tabcolsep}{2pt}
	\begin{tabularx}{\columnwidth}{@{}lXlX@{}}
		\toprule
		\textbf{Symbol} & \textbf{Meaning}
		& \textbf{Symbol} & \textbf{Meaning} \\
		\midrule
		
		\(a\), \(a^\ast\)
		& Order / optimum
		& \(\mathbf{F}^{a}\)
		& GFRFT matrix \\
		
		\(\mathbf{D}_s\)
		& Downsampling operator
		& \(\mathbf{\Phi}_a\)
		& Predefined reconstruction \\
		
		\(\mathbf{S}_a^{\ast}\)
		& Sampling operator
		& \(\mathbf{H}_a(\mathbf{G})\)
		& Correction operator \\
		
		\(\mathbf{W}_a\)
		& Reconstruction operator
		& \(\boldsymbol{\Gamma}_x\), \(\boldsymbol{\Gamma}_{\eta}\)
		& Signal / noise covariance \\
		
		\(\mathcal{U}_{K,a}\)
		& Fractional subspace
		& \(\mathbf{P}_{K,a}\)
		& Projection operator \\
		
		\(\rho_K(a)\)
		& Energy concentration
		& \(\mathbf{r}_{K,a}\)
		& Projection residual \\
		
		\(\mathbf{A}_a\)
		& Effective sampling matrix
		& \(\tau_a\), \(\nu_a\)
		& Leakage / noise terms \\
		
		SS, SM, ST
		& Subspace, smoothness, stochastic priors
		& UNC, PD
		& Unconstrained and predefined recovery \\
		
		\bottomrule
	\end{tabularx}
\end{table}
\subsection{Subspace Prior}
\label{subsec:subspace_prior_framework}

Under the subspace prior introduced in Section~II, the GFRFT domain sampling reconstruction framework can be specialized by substituting the low-dimensional representation \(\mathbf{x}=\mathbf{M}\mathbf{d}\) into the sampling model. The sampled signal becomes
\begin{equation}
	\mathbf{y}_a
	=
	\mathbf{S}_a^{\ast}
	\mathbf{M}
	\mathbf{d}.
	\label{eq:subspace_sampled_signal}
\end{equation}

\begin{proposition}
	\label{prop:subspace_based_reconstruction}
	Assume that \(\mathbf{S}_a^{\ast}\mathbf{M}\) has full column rank. Then the unconstrained reconstruction under the subspace prior is given by
	\begin{equation}
		\widetilde{\mathbf{x}}_{a,\mathrm{SS}}^{\mathrm{UNC}}
		=
		\mathbf{M}
		\left(
		\mathbf{S}_a^{\ast}
		\mathbf{M}
		\right)^{\dagger}
		\mathbf{y}_a .
		\label{eq:subspace_reconstruction}
	\end{equation}
	Equivalently, the corresponding reconstruction and correction operators are
	\begin{equation}
		\mathbf{W}_{a,\mathrm{SS}}^{\mathrm{UNC}}
		=
		\mathbf{M},
		\qquad
		\mathbf{H}_{a,\mathrm{SS}}^{\mathrm{UNC}}
		=
		\left(
		\mathbf{S}_a^{\ast}
		\mathbf{M}
		\right)^{\dagger}.
		\label{eq:subspace_unc_hw}
	\end{equation}
	
	For predefined reconstruction, let \(\boldsymbol{\Phi}_a\in\mathbb{C}^{N\times K}\) be a prescribed reconstruction matrix and assume that
	\(\boldsymbol{\Phi}_a^{\mathrm{H}}\boldsymbol{\Phi}_a\) is nonsingular. Then the predefined correction operator is
	\begin{equation}
		\mathbf{H}_{a,\mathrm{SS}}^{\mathrm{PD}}
		=
		\left(
		\boldsymbol{\Phi}_a^{\mathrm{H}}
		\boldsymbol{\Phi}_a
		\right)^{-1}
		\boldsymbol{\Phi}_a^{\mathrm{H}}
		\mathbf{M}
		\left(
		\mathbf{S}_a^{\ast}
		\mathbf{M}
		\right)^{\dagger},
		\label{eq:subspace_pre_h}
	\end{equation}
	and the corresponding reconstruction is
	\begin{equation}
		\widetilde{\mathbf{x}}_{a,\mathrm{SS}}^{\mathrm{PD}}
		=
		\boldsymbol{\Phi}_a
		\mathbf{H}_{a,\mathrm{SS}}^{\mathrm{PD}}
		\mathbf{y}_a .
		\label{eq:subspace_pre_reconstruction}
	\end{equation}
\end{proposition}

\textit{Proof.} See Appendix~\ref{app:proof_propositions}.

\subsection{Smoothness Prior}
\label{subsec:smoothness_prior_framework}

Under the smoothness prior introduced in Section~II, we specify its GFRFT domain realization. This prior constrains the reconstructed signal through a variation measuring operator. For a given fractional order \(a\), the same smoothness response rule is applied in the \(a\) order GFRFT domain, leading to
\begin{equation}
	\mathbf{T}_a(\mathbf{G})
	=
	\operatorname{diag}
	\left(
	T_a(g_0),
	T_a(g_1),
	\ldots,
	T_a(g_{N-1})
	\right),
	\label{eq:smoothness_filter_response_matrix}
\end{equation}
and the corresponding vertex domain operator is
\begin{equation}
	\mathbf{T}_{v,a}
	=
	\mathbf{F}^{-a}
	\mathbf{T}_a(\mathbf{G})
	\mathbf{F}^{a}.
	\label{eq:smoothness_vertex_operator}
\end{equation}
The smoothness constraint is then imposed as
\begin{equation}
	\left\|
	\mathbf{T}_{v,a}
	\mathbf{x}
	\right\|_2^2
	\leq
	\rho^2 .
	\label{eq:smoothness_prior_constraint}
\end{equation}

Based on this constraint, the smoothness induced reconstruction operator is defined as
\begin{equation}
	\mathbf{R}_{a,\mathrm{SM}}
	=
	\left(
	\mathbf{T}_{v,a}^{\mathrm{H}}
	\mathbf{T}_{v,a}
	\right)^{-1}
	\left(
	\mathbf{S}_a^{\ast}
	\right)^{\mathrm{H}} .
	\label{eq:smoothness_induced_reconstruction_operator}
\end{equation}

\begin{proposition}
	\label{prop:smoothness_based_reconstruction}
	Assume that \(\mathbf{S}_a^{\ast}\mathbf{R}_{a,\mathrm{SM}}\) has full column rank. Then the unconstrained reconstruction under the smoothness prior is
	\begin{equation}
		\widetilde{\mathbf{x}}_{a,\mathrm{SM}}^{\mathrm{UNC}}
		=
		\mathbf{R}_{a,\mathrm{SM}}
		\left(
		\mathbf{S}_a^{\ast}
		\mathbf{R}_{a,\mathrm{SM}}
		\right)^{\dagger}
		\mathbf{y}_a .
		\label{eq:smoothness_unc_reconstruction}
	\end{equation}
	Equivalently, the reconstruction and correction operators are
	\begin{equation}
		\mathbf{W}_{a,\mathrm{SM}}^{\mathrm{UNC}}
		=
		\mathbf{R}_{a,\mathrm{SM}},
		\qquad
		\mathbf{H}_{a,\mathrm{SM}}^{\mathrm{UNC}}
		=
		\left(
		\mathbf{S}_a^{\ast}
		\mathbf{R}_{a,\mathrm{SM}}
		\right)^{\dagger}.
		\label{eq:smoothness_unc_hw}
	\end{equation}
	
	For predefined reconstruction, let \(\boldsymbol{\Phi}_a\in\mathbb{C}^{N\times K}\) be a prescribed reconstruction matrix with nonsingular
	\(\boldsymbol{\Phi}_a^{\mathrm{H}}\boldsymbol{\Phi}_a\). Then
	\begin{equation}
		\mathbf{H}_{a,\mathrm{SM}}^{\mathrm{PD}}
		=
		\left(
		\boldsymbol{\Phi}_a^{\mathrm{H}}
		\boldsymbol{\Phi}_a
		\right)^{-1}
		\boldsymbol{\Phi}_a^{\mathrm{H}}
		\mathbf{R}_{a,\mathrm{SM}}
		\left(
		\mathbf{S}_a^{\ast}
		\mathbf{R}_{a,\mathrm{SM}}
		\right)^{\dagger},
		\label{eq:smoothness_pre_h}
	\end{equation}
	and the corresponding reconstruction is
	\begin{equation}
		\widetilde{\mathbf{x}}_{a,\mathrm{SM}}^{\mathrm{PD}}
		=
		\boldsymbol{\Phi}_a
		\mathbf{H}_{a,\mathrm{SM}}^{\mathrm{PD}}
		\mathbf{y}_a .
		\label{eq:smoothness_pre_reconstruction}
	\end{equation}
\end{proposition}

 \textit{Proof.} See Appendix~\ref{app:proof_propositions}.

\subsection{Stochastic Prior}
\label{subsec:stochastic_prior_framework}

Under the stochastic prior introduced in Section~II, we incorporate the second order statistical information into the proposed GFRFT domain sampling reconstruction framework. Let \(\boldsymbol{\Gamma}_x\) and \(\boldsymbol{\Gamma}_{\eta}\) denote the covariance matrices of the zero-mean graph signal \(\mathbf{x}\) and the sampling noise \(\boldsymbol{\eta}\), respectively:
\begin{equation}
	\boldsymbol{\Gamma}_x
	=
	\mathbb{E}
	\left[
	\mathbf{x}
	\mathbf{x}^{\mathrm{H}}
	\right],
	\qquad
	\boldsymbol{\Gamma}_{\eta}
	=
	\mathbb{E}
	\left[
	\boldsymbol{\eta}
	\boldsymbol{\eta}^{\mathrm{H}}
	\right].
	\label{eq:stochastic_covariance_matrices}
\end{equation}
Assume that \(\mathbf{x}\) and \(\boldsymbol{\eta}\) are independent. The following proposition gives the corresponding stochastic prior based reconstruction form.

\begin{proposition}
	\label{prop:stochastic_reconstruction}
	Under the stochastic prior, the unconstrained linear MSE optimal reconstruction is
	\begin{equation}
		\widetilde{\mathbf{x}}_{a,\mathrm{ST}}^{\mathrm{UNC}}
		=
		\boldsymbol{\Gamma}_x
		\left(
		\mathbf{S}_a^{\ast}
		\right)^{\mathrm{H}}
		\left(
		\mathbf{S}_a^{\ast}
		\boldsymbol{\Gamma}_x
		\left(
		\mathbf{S}_a^{\ast}
		\right)^{\mathrm{H}}
		+
		\boldsymbol{\Gamma}_{\eta}
		\right)^{-1}
		\mathbf{y}_a .
		\label{eq:stochastic_unc_reconstruction}
	\end{equation}
	Equivalently, the corresponding reconstruction and correction operators are
	\begin{equation}
		\mathbf{W}_{a,\mathrm{ST}}^{\mathrm{UNC}}
		=
		\boldsymbol{\Gamma}_x
		\left(
		\mathbf{S}_a^{\ast}
		\right)^{\mathrm{H}},
		\qquad
		\mathbf{H}_{a,\mathrm{ST}}^{\mathrm{UNC}}
		=
		\left(
		\mathbf{S}_a^{\ast}
		\boldsymbol{\Gamma}_x
		\left(
		\mathbf{S}_a^{\ast}
		\right)^{\mathrm{H}}
		+
		\boldsymbol{\Gamma}_{\eta}
		\right)^{-1}.
		\label{eq:stochastic_unc_hw}
	\end{equation}
	
	For predefined reconstruction, let \(\boldsymbol{\Phi}_a\in\mathbb{C}^{N\times K}\) be a prescribed reconstruction matrix with nonsingular
	\(\boldsymbol{\Phi}_a^{\mathrm{H}}\boldsymbol{\Phi}_a\). Then
	\begin{equation}
		\mathbf{H}_{a,\mathrm{ST}}^{\mathrm{PD}}
		=
		\left(
		\boldsymbol{\Phi}_a^{\mathrm{H}}
		\boldsymbol{\Phi}_a
		\right)^{-1}
		\boldsymbol{\Phi}_a^{\mathrm{H}}
		\boldsymbol{\Gamma}_x
		\left(
		\mathbf{S}_a^{\ast}
		\right)^{\mathrm{H}}
		\left(
		\mathbf{S}_a^{\ast}
		\boldsymbol{\Gamma}_x
		\left(
		\mathbf{S}_a^{\ast}
		\right)^{\mathrm{H}}
		+
		\boldsymbol{\Gamma}_{\eta}
		\right)^{-1},
		\label{eq:stochastic_pre_h}
	\end{equation}
	and the corresponding reconstruction is
	\begin{equation}
		\widetilde{\mathbf{x}}_{a,\mathrm{ST}}^{\mathrm{PD}}
		=
		\boldsymbol{\Phi}_a
		\mathbf{H}_{a,\mathrm{ST}}^{\mathrm{PD}}
		\mathbf{y}_a .
		\label{eq:stochastic_pre_reconstruction}
	\end{equation}
\end{proposition}

\textit{Proof.} See Appendix~\ref{app:proof_propositions}.

\begin{figure*}[!t]
	\centering
	\includegraphics[width=0.80\textwidth]{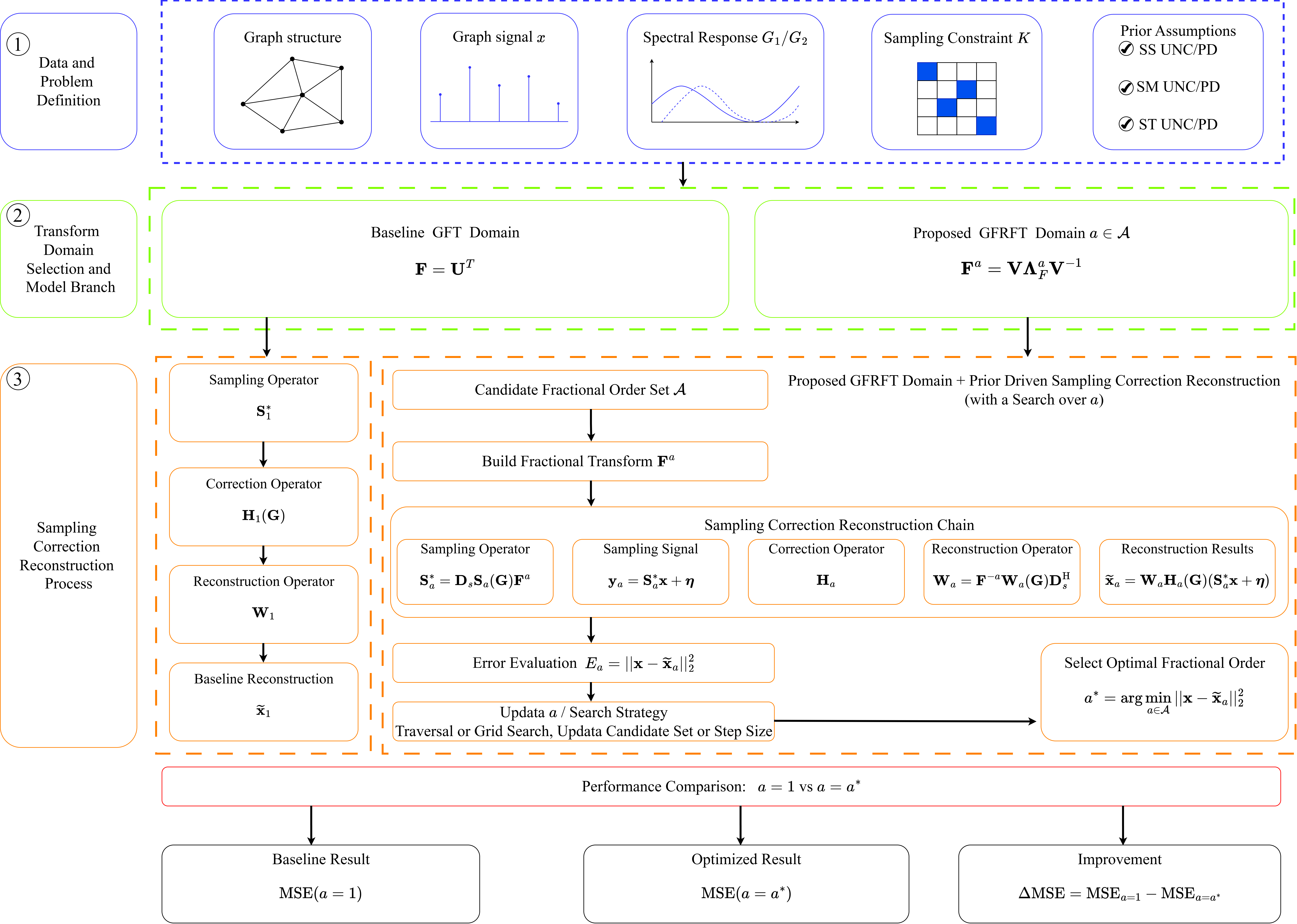}
	
	\vspace{1mm}
	\refstepcounter{figure}
	\label{fig:gfrft_framework}
	{\footnotesize Fig.~\thefigure. Overall framework of the GFRFT domain sampling reconstruction method.\par}
\end{figure*}

\subsection{Optimal Fractional Order Selection}
\label{subsec:optimal_fractional_order}

The optimal fractional order is selected by minimizing the reconstruction error:
\begin{equation}
	a^{\ast}
	=
	\arg\min_{a\in\mathcal{A}}
	\left\|
	\mathbf{x}
	-
	\widetilde{\mathbf{x}}_a
	\right\|_2^2 .
	\label{eq:optimal_fractional_order}
\end{equation}

The overall framework of the GFRFT domain sampling reconstruction method is illustrated in Fig.~\ref{fig:gfrft_framework}, where the GFT domain reconstruction branch and the fractional domain order search branch are compared within a unified reconstruction pipeline. The overall procedure for the GFRFT domain sampling reconstruction framework is summarized in Algorithm~\ref{alg:gfrft_order_selection}.

In supervised experiments, the ground truth graph signal is available, and the reconstruction error can be directly evaluated. This fractional order search can be interpreted as adaptive spectral domain selection. If $a^{\ast}=1$, the conventional GFT domain is already suitable for the graph signal. If $a^{\ast}\neq1$, a fractional spectral domain provides a more suitable representation for the given signal and sampling model.

\begin{algorithm}[!t]
	\caption{Fractional order selection for GFRFT domain reconstruction}
	\label{alg:gfrft_order_selection}
	\begin{algorithmic}[1]
		\Require Graph signal $\mathbf{x}$, candidate order set $\mathcal{A}$, sampling number $K$, prior type $\mathcal{P}$, and reconstruction mode $\mathcal{M}$
		\Ensure Optimal fractional order $a^{\ast}$ and reconstructed signal $\widetilde{\mathbf{x}}_{a^{\ast}}$
		
		\State Initialize $E_{\min}\leftarrow +\infty$, $a^{\ast}\leftarrow 1$
		
		\For{each $a\in\mathcal{A}$}
		\State Construct $\mathbf{F}^{a}$ and compute the GFRFT coefficients
		\[
		\widehat{\mathbf{x}}_a=\mathbf{F}^{a}\mathbf{x}
		\]
		
		\State Obtain the sampled signal by spectral filtering and folding
		\[
		\mathbf{y}_a
		=
		\mathbf{S}_a^{\ast}\mathbf{x}
		+
		\boldsymbol{\eta}
		\]
		
		\State Design the correction operator $\mathbf{H}_a(\mathbf{G})$ and the reconstruction operator $\mathbf{W}_a$ according to $\mathcal{P}$ and $\mathcal{M}$
		
		\State Reconstruct the graph signal
		\[
		\widetilde{\mathbf{x}}_a
		=
		\mathbf{W}_a
		\mathbf{H}_a(\mathbf{G})
		\mathbf{y}_a
		\]
		
		\State Compute the reconstruction error
		\[
		E_a
		=
		10\log_{10}
		\left(
		\frac{1}{N}
		\left\|
		\mathbf{x}
		-
		\widetilde{\mathbf{x}}_a
		\right\|_2^2
		\right)
		\]
		
		\If{$E_a<E_{\min}$}
		\State $E_{\min}\leftarrow E_a$
		\State $a^{\ast}\leftarrow a$
		\State $\widetilde{\mathbf{x}}_{a^{\ast}}\leftarrow \widetilde{\mathbf{x}}_a$
		\EndIf
		\EndFor
		
		\State \Return $a^{\ast}$ and $\widetilde{\mathbf{x}}_{a^{\ast}}$
	\end{algorithmic}
\end{algorithm}

\section{Theoretical Interpretation of GFRFT Domain Reconstruction}
\label{sec:gfrft_projection_mechanism}

\subsection{Fractional Spectral Subspace and Energy Concentration}
\label{subsec:fractional_subspace_energy_concentration}

The preceding section formulates graph signal sampling and reconstruction in the GFRFT domain. This section further explains why selecting a suitable fractional order may improve reconstruction performance. The key point is that the fractional order \(a\) changes the spectral representation used for low-dimensional approximation. When \(a=1\), the representation reduces to the GFT domain, and the reconstructed signal is constrained by a GFT spectral subspace.

For a fixed subspace dimension \(K\), different fractional orders generally correspond to different \(K\)-dimensional spectral subspaces. Therefore, the role of \(a\) is not merely to rescale spectral coefficients, but to determine which spectral subspace is used to capture the dominant structure of the signal. If a certain fractional subspace is better aligned with the signal structure, more energy can be captured with the same dimension \(K\), and the residual outside the selected subspace can be reduced.

\begin{definition}
	\label{def:fractional_spectral_energy_concentration}
	Let \(\mathbf{F}^{a}\) denote the \(a\) order GFRFT matrix. For a given subspace dimension \(K\), let \(\mathcal{U}_{K,a}\) be the \(K\)-dimensional fractional spectral subspace spanned by the selected dominant spectral components in the \(a\) order GFRFT domain. Denote by \(\mathbf{U}_{K,a}\) an orthonormal basis of \(\mathcal{U}_{K,a}\). The orthogonal projection operator onto this subspace is defined as
	\begin{equation}
		\mathbf{P}_{K,a}
		=
		\mathbf{U}_{K,a}\mathbf{U}_{K,a}^{\mathrm{H}} .
		\label{eq:projection_operator_fractional}
	\end{equation}
	Then, any graph signal \(\mathbf{x}\) can be decomposed as
	\begin{equation}
		\mathbf{x}
		=
		\mathbf{x}_{K,a}
		+
		\mathbf{r}_{K,a},
		\label{eq:fractional_low_dimensional_decomposition}
	\end{equation}
	where
	\begin{equation}
		\mathbf{x}_{K,a}
		=
		\mathbf{P}_{K,a}\mathbf{x}
	\end{equation}
	is the component captured by \(\mathcal{U}_{K,a}\), and
	\begin{equation}
		\mathbf{r}_{K,a}
		=
		\left(
		\mathbf{I}
		-
		\mathbf{P}_{K,a}
		\right)
		\mathbf{x}
	\end{equation}
	is the residual component outside this subspace.
	
	The \(K\)-dimensional energy concentration ratio in the \(a\)-order fractional spectral domain is defined as
	\begin{equation}
		\rho_K(a)
		=
		\frac{
			\left\|
			\mathbf{P}_{K,a}\mathbf{x}
			\right\|_2^2
		}{
			\left\|
			\mathbf{x}
			\right\|_2^2
		}.
		\label{eq:fractional_energy_concentration_ratio}
	\end{equation}
\end{definition}

According to Definition~\ref{def:fractional_spectral_energy_concentration}, \(\rho_K(a)\) measures the proportion of signal energy captured by the \(K\)-dimensional fractional spectral subspace. Since \(\mathbf{P}_{K,a}\) is an orthogonal projection operator, the residual energy satisfies
\begin{equation}
	\left\|
	\mathbf{r}_{K,a}
	\right\|_2^2
	=
	\left[
	1-\rho_K(a)
	\right]
	\left\|
	\mathbf{x}
	\right\|_2^2 .
	\label{eq:residual_energy_concentration_relation}
\end{equation}
Thus, a larger \(\rho_K(a)\) directly implies a smaller low-dimensional projection residual. In particular, if there exists a fractional order \(a\neq 1\) such that
\begin{equation}
	\rho_K(a)
	>
	\rho_K(1),
	\label{eq:fractional_energy_advantage_condition}
\end{equation}
then the corresponding GFRFT domain captures more signal energy than the GFT domain under the same subspace dimension \(K\). This indicates that the GFRFT domain can provide a more suitable low-dimensional spectral subspace for the given graph signal, thereby reducing the projection residual and potentially improving the sampling and reconstruction performance.

\subsection{From Fractional Low-Dimensional Representation to Sampling Reconstruction}
\label{subsec:fractional_low_dimensional_sampling}

The preceding analysis shows that the fractional order \(a\) can change the low-dimensional spectral subspace and hence affect the projection residual of a graph signal. However, a smaller projection residual does not necessarily imply a smaller reconstruction error in a sampling based setting. This is because the final reconstruction performance also depends on the identifiability of the sampled subspace, the leakage of the spectral residual, and the amplification of sampling noise.

Substituting the low-dimensional decomposition in
\eqref{eq:fractional_low_dimensional_decomposition} into the sampling model
in \eqref{eq:gfrft_sampled_signal}, the sampled observation can be rewritten as
\begin{equation}
	\mathbf{y}_a
	=
	\mathbf{S}_a^{\ast}\mathbf{x}_{K,a}
	+
	\mathbf{S}_a^{\ast}\mathbf{r}_{K,a}
	+
	\boldsymbol{\eta}.
	\label{eq:fractional_sampling_decomposition}
\end{equation}
Since \(\mathbf{x}_{K,a}\in\mathcal{U}_{K,a}\), there exists a coefficient vector
\(\mathbf{c}_a\) such that
\begin{equation}
	\mathbf{x}_{K,a}
	=
	\mathbf{U}_{K,a}\mathbf{c}_a .
	\label{eq:low_dimensional_coefficient}
\end{equation}
By defining the effective sampling matrix as
\begin{equation}
	\mathbf{A}_a
	=
	\mathbf{S}_a^{\ast}\mathbf{U}_{K,a},
	\label{eq:fractional_effective_sampling_matrix}
\end{equation}
the observation model becomes
\begin{equation}
	\mathbf{y}_a
	=
	\mathbf{A}_a\mathbf{c}_a
	+
	\mathbf{S}_a^{\ast}\mathbf{r}_{K,a}
	+
	\boldsymbol{\eta}.
	\label{eq:fractional_sampling_observation}
\end{equation}
Equation~\eqref{eq:fractional_sampling_observation} shows that the sampled signal consists of the low-dimensional component, the sampled residual component, and the noise term. Therefore, the reconstruction performance depends not only on the magnitude of the projection residual \(\mathbf{r}_{K,a}\), but also on the stability of the effective sampling matrix \(\mathbf{A}_a\).

If \(\mathbf{A}_a\) has full column rank, the low-dimensional coefficient can be estimated by the least squares solution
\begin{equation}
	\widehat{\mathbf{c}}_a
	=
	\arg\min_{\mathbf{c}}
	\left\|
	\mathbf{y}_a-\mathbf{A}_a\mathbf{c}
	\right\|_2^2
	=
	\mathbf{A}_a^{\dagger}\mathbf{y}_a ,
	\label{eq:fractional_ls_coefficient}
\end{equation}
and the corresponding reconstruction is
\begin{equation}
	\widetilde{\mathbf{x}}_a
	=
	\mathbf{U}_{K,a}\widehat{\mathbf{c}}_a
	=
	\mathbf{U}_{K,a}\mathbf{A}_a^{\dagger}\mathbf{y}_a .
	\label{eq:fractional_ls_reconstruction}
\end{equation}
This formulation indicates that the reconstruction error is affected by three factors: the energy concentration of the fractional subspace, the leakage of the residual component through sampling and inversion, and the amplification of the noise by \(\mathbf{A}_a^{\dagger}\).

\begin{proposition}
	\label{prop:fractional_projection_advantage}
	If there exists a fractional order \(a\neq 1\) such that
	\begin{equation}
		\rho_K(a)>\rho_K(1),
		\label{eq:prop_energy_condition}
	\end{equation}
	then
	\begin{equation}
		\left\|
		\mathbf{r}_{K,a}
		\right\|_2^2
		<
		\left\|
		\mathbf{r}_{K,1}
		\right\|_2^2 .
		\label{eq:prop_residual_condition}
	\end{equation}
\end{proposition}

\begin{proof}[\hspace*{1em}Proof]
	Since \(\rho_K(a)>\rho_K(1)\) implies
	\(1-\rho_K(a)<1-\rho_K(1)\), the conclusion follows directly from
	\eqref{eq:residual_energy_concentration_relation}.
\end{proof}

Proposition~\ref{prop:fractional_projection_advantage} shows that a fractional spectral subspace with a higher energy concentration ratio has a smaller low-dimensional projection residual than the GFT subspace. Nevertheless, this condition alone is not sufficient to guarantee a smaller sampling reconstruction error, because the residual leakage and the noise amplification must also be controlled. The following theorem gives a sufficient condition under which the fractional domain sampling reconstruction outperforms the GFT domain reconstruction.

\begin{theorem}
	\label{thm:fractional_reconstruction_advantage}
	Consider the fractional domain sampling model in
	\eqref{eq:gfrft_sampled_signal} and the low-dimensional decomposition in
	\eqref{eq:fractional_low_dimensional_decomposition}. Let
	\(\widetilde{\mathbf{x}}_a\) be the least squares reconstruction defined in
	\eqref{eq:fractional_ls_reconstruction}. Assume that the effective sampling matrix \(\mathbf{A}_a\) defined in
	\eqref{eq:fractional_effective_sampling_matrix} has full column rank, and that the sampling noise satisfies
	\begin{equation}
		\mathbb{E}\left[\boldsymbol{\eta}\right]
		=
		\mathbf{0},
		\quad
		\mathbb{E}
		\left[
		\boldsymbol{\eta}\boldsymbol{\eta}^{\mathrm{H}}
		\right]
		=
		\sigma^2\mathbf{I}.
		\label{eq:theorem_noise_assumption}
	\end{equation}
	Define the residual leakage coefficient and the noise amplification term as
	\begin{equation}
		\tau_a
		=
		\left\|
		\mathbf{A}_a^{\dagger}
		\mathbf{S}_a^{\ast}
		\left(
		\mathbf{I}-\mathbf{P}_{K,a}
		\right)
		\right\|_2
		\label{eq:theorem_leakage_coefficient}
	\end{equation}
	and
	\begin{equation}
		\nu_a
		=
		\sigma^2
		\left\|
		\mathbf{A}_a^{\dagger}
		\right\|_F^2 .
		\label{eq:theorem_noise_amplification}
	\end{equation}
	If there exists a fractional order \(a\in(0,1)\) such that
	\begin{equation}
		\left(1+\tau_a^2\right)
		\left\|
		\mathbf{r}_{K,a}
		\right\|_2^2
		+
		\nu_a
		<
		\left\|
		\mathbf{r}_{K,1}
		\right\|_2^2,
		\label{eq:theorem_sufficient_condition}
	\end{equation}
	then
	\begin{equation}
		\mathbb{E}
		\left[
		\left\|
		\mathbf{x}
		-
		\widetilde{\mathbf{x}}_a
		\right\|_2^2
		\right]
		<
		\mathbb{E}
		\left[
		\left\|
		\mathbf{x}
		-
		\widetilde{\mathbf{x}}_1
		\right\|_2^2
		\right],
		\label{eq:theorem_reconstruction_advantage}
	\end{equation}
	where \(\widetilde{\mathbf{x}}_1\) denotes the reconstruction constrained to
	the \(K\)-dimensional GFT subspace.
\end{theorem}

\textit{Proof.} See Appendix~\ref{app:proof_fractional_reconstruction_advantage}. 

The theorem~\ref{thm:fractional_reconstruction_advantage} states that the GFRFT domain reconstruction is not always superior to the GFT domain reconstruction. Instead, the energy concentration gain obtained by the fractional subspace must be sufficiently large to compensate for the additional residual leakage and noise amplification caused by the sampling inversion. Hence, the theoretical advantage of the GFRFT domain is achieved when a more suitable fractional subspace is selected and the corresponding sampled subspace remains stably invertible.

\subsection{Prior Dependent Interpretation of the Reconstruction Condition}
\label{subsec:prior_dependent_reconstruction_condition}

Theorem~\ref{thm:fractional_reconstruction_advantage} shows that the reconstruction advantage of the GFRFT domain is governed by three terms: the projection residual \(\|\mathbf{r}_{K,a}\|_2^2\), the residual leakage coefficient \(\tau_a\), and the noise amplification term \(\nu_a\). Equivalently, the fractional domain is favored when
\begin{equation}
	\mathcal{E}_a
	=
	\left(
	1+\tau_a^2
	\right)
	\left\|
	\mathbf{r}_{K,a}
	\right\|_2^2
	+
	\nu_a
	\label{eq:prior_dependent_error_indicator}
\end{equation}
is sufficiently smaller than the GFT domain reconstruction error. Therefore, different priors can be interpreted according to which component of \(\mathcal{E}_a\) they mainly affect.

Under the subspace prior, the signal is assumed to be mainly represented by a low-dimensional spectral subspace. This prior directly acts on the projection residual \(\|\mathbf{r}_{K,a}\|_2^2\). If the fractional subspace \(\mathcal{U}_{K,a}\) better matches the intrinsic signal subspace than the GFT subspace \(\mathcal{U}_{K,1}\), then
\begin{equation*}
	\left\|
	\mathbf{r}_{K,a}
	\right\|_2^2
	<
	\left\|
	\mathbf{r}_{K,1}
	\right\|_2^2 .
	\label{eq:subspace_prior_residual_reduction}
\end{equation*}
Thus, the subspace prior mainly reduces the approximation error caused by low-dimensional truncation. However, this reduction improves the final reconstruction only when the effective sampling matrix remains stable, so that \(\tau_a\) and \(\nu_a\) are not significantly amplified.

Under the smoothness prior, the signal is encouraged to vary smoothly over the graph. This prior affects both the projection residual and the residual leakage. An appropriate fractional order can concentrate the smooth component into fewer dominant fractional spectral components, thereby reducing \(\|\mathbf{r}_{K,a}\|_2^2\). Meanwhile, the residual component outside the selected subspace usually corresponds to higher graph variations. If such components are suppressed more effectively in the fractional domain, their leakage after sampling and inversion can also be reduced. Hence, the smoothness prior mainly controls the combined term
\begin{equation*}
	\left(
	1+\tau_a^2
	\right)
	\left\|
	\mathbf{r}_{K,a}
	\right\|_2^2 .
	\label{eq:smoothness_prior_controlled_term}
\end{equation*}

Under the stochastic prior, the graph signal is modeled as a random vector. In this case, the representation error is evaluated in an average sense. The expected residual energy outside the fractional subspace is
\begin{equation*}
	\mathbb{E}
	\left[
	\left\|
	\mathbf{r}_{K,a}
	\right\|_2^2
	\right]
	=
	\operatorname{tr}
	\left[
	\left(
	\mathbf{I}
	-
	\mathbf{P}_{K,a}
	\right)
	\boldsymbol{\Gamma}_x
	\right] .
	\label{eq:stochastic_prior_expected_residual}
\end{equation*}
Therefore, the stochastic prior mainly reduces the expected projection residual by selecting a fractional subspace that captures the dominant covariance energy. When statistical reconstruction is further considered, the covariance information can also help balance signal recovery and noise suppression, thereby affecting the practical influence of \(\nu_a\).

In summary, the three priors explain the reconstruction condition from different perspectives. The subspace prior emphasizes low-dimensional approximation, the smoothness prior emphasizes the suppression of high variation residual components, and the stochastic prior emphasizes average energy concentration and noise robustness. Therefore, the GFRFT domain advantage depends on whether the selected fractional subspace is consistent with the signal prior and remains stable under the sampling operator.

\section{Numerical experiments}
\subsection{Experimental Settings}

This section evaluates the effectiveness of the GFRFT domain sampling and reconstruction framework using both simulated graph signals and real traffic sensor graph signals. The main comparison is conducted between two reconstruction settings. The first setting fixes the fractional order at \(a=1\), which corresponds to the GFT domain. The second setting searches for the optimal fractional order \(a^\ast\) from a predefined candidate set and performs reconstruction in the corresponding GFRFT domain. By comparing the reconstruction errors under these two settings, we examine whether the adjustable fractional spectral domain can improve graph signal reconstruction performance. 

Two spectral response functions are used in the experiments:
\begin{equation}
	G_1(\lambda)=1-\frac{\lambda}{\lambda_{\max}+\varepsilon},
\end{equation}
\begin{equation}
	G_2(\lambda)=\exp\left(-\frac{\lambda}{\lambda_{\max}}\right),
\end{equation}
where \(\lambda_{\max}\) denotes the largest eigenvalue of the graph Laplacian, and \(\varepsilon=0.3\). The function \(G_1\) is a linearly decaying response, while \(G_2\) is an exponentially decaying response. Both functions are used to characterize spectral weighting in the sampling and reconstruction processes.

In the predefined recovery variants, an additional predefined spectral reconstruction response is introduced as
\begin{equation}
	W(\lambda_i)
	=
	\cos\left(
	\frac{\lambda_i}{\lambda_{\max}+\varepsilon}
	\right),
	\quad i=1,\ldots,N .
\end{equation}
This response is used as a fixed spectral weighting function in the PD reconstruction methods. 

The reconstruction error is measured by the logarithmic mean-squared error:
\begin{equation}
	\mathrm{MSE}_{\mathrm{dB}}
	=
	10\log_{10}
	\left(
	\frac{1}{N}
	\sum_{i=1}^{N}
	|x_i-\hat{x}_i|^2
	\right),
\end{equation}
where \(x_i\) denotes the original graph signal value at the \(i\)-th node, and \(\hat{x}_i\) denotes the corresponding reconstructed value. A lower \(\mathrm{MSE}_{\mathrm{dB}}\) indicates a smaller reconstruction error. 
To quantify the benefit of optimizing the fractional order, the reconstruction improvement is defined as
\begin{equation}
	\Delta \mathrm{MSE}
	=
	\mathrm{MSE}_{a=1}
	-
	\mathrm{MSE}_{a=a^\ast}.
\end{equation}
Here, \(\mathrm{MSE}_{a=1}\) denotes the reconstruction error in the GFT domain, and \(\mathrm{MSE}_{a=a^\ast}\) denotes the minimum error obtained over the searched fractional orders. \(\Delta \mathrm{MSE}\) quantifies the error reduction achieved by the optimized GFRFT domain.

The fractional order is searched over the following candidate set:
\begin{equation*}
	a\in\{0,0.02,0.04,\ldots,1.00\}.
\end{equation*}
The case \(a=1\) is used as the GFT domain baseline, while the optimal order \(a^\ast\) corresponds to the best reconstruction result in the GFRFT domain. To ensure reproducibility, all experiments are conducted with fixed random seeds.

\subsection{Simulated Experiments}

\subsubsection{Simulated Graph Structure and Signal Construction}

The simulated experiments are conducted on a sensor graph. Each node represents a sensor located in a two dimensional space, and the edges are determined according to spatial proximity. The number of graph nodes is set to $N=128$.

Two types of simulated graph signals are used in the experiments:
\begin{itemize}
	\item \textbf{Near-subspace signal:} A signal whose energy is mainly concentrated on a few spectral components, with a small spectral tail retained.
	
	\item \textbf{Smooth-local-anomaly signal:} A signal composed of a smooth graph component and several localized anomalies.
\end{itemize}

The simulated experiments are divided into noiseless and noisy cases. In the noiseless case, the clean graph signals are directly sampled and reconstructed. In the noisy case, vertex domain Gaussian white noise is added to the clean signal:
\begin{equation*}
	x_{\mathrm{noisy}}=x_{\mathrm{clean}}+n,
\end{equation*}
\begin{equation*}
	n\sim \mathcal{N}(0,\sigma^2I).
\end{equation*}
The noise standard deviation is set to \(\sigma=0.10\).

Fig.~\ref{fig:simulated_signal_visualization} visualizes the two simulated graph signals under noiseless and noisy settings. The first two panels correspond to the near-subspace signal before and after adding noise, while the last two panels correspond to the smooth-local-anomaly signal before and after adding noise.
\begin{figure*}[!t]
	\centering
	\setlength{\tabcolsep}{1.2pt}
	\renewcommand{\arraystretch}{0.95}
	
	\begin{tabular}{cccc}
		\includegraphics[width=0.235\textwidth]{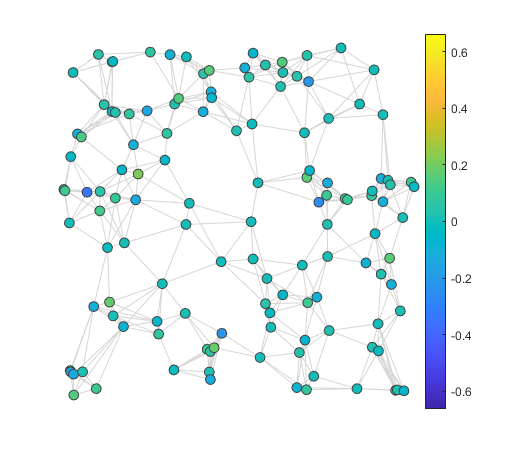} &
		\includegraphics[width=0.235\textwidth]{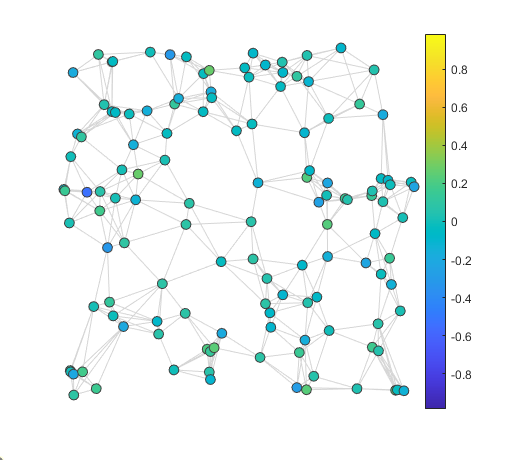} &
		\includegraphics[width=0.235\textwidth]{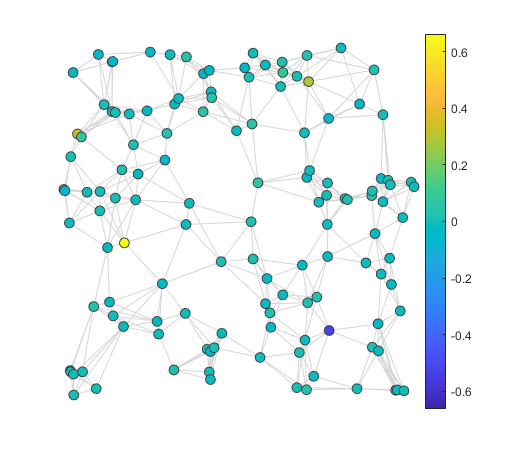} &
		\includegraphics[width=0.235\textwidth]{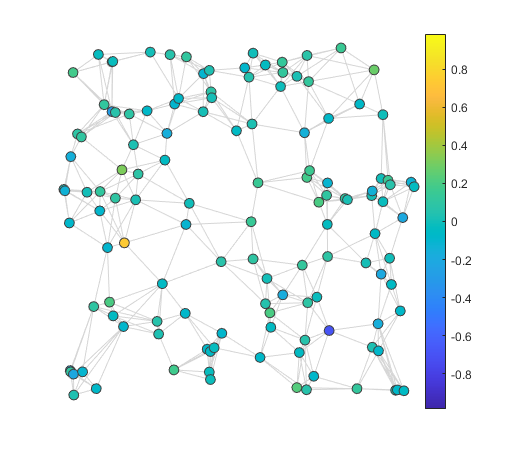} \\
		\footnotesize (a) Near-subspace &
		\footnotesize (b) Noisy near-subspace &
		\footnotesize (c) Smooth-local &
		\footnotesize (d) Noisy smooth-local
	\end{tabular}
	
	\vspace{1mm}
	\refstepcounter{figure}
	\label{fig:simulated_signal_visualization}
	{\centering\footnotesize Fig.~\thefigure. Visualization of the simulated graph signals under noiseless and noisy settings.\par}
\end{figure*}

\subsubsection{Results on Noiseless Simulated Signals}

The noiseless experiments are conducted under two spectral response functions, \(G_1\) and \(G_2\). For each response function, both the near-subspace signal and the smooth-local-anomaly mixed signal are tested. 

Table~\ref{tab:noiseless_simulated_summary} summarizes the reconstruction errors of different methods under the noiseless simulated settings. The results show that, for most reconstruction methods, the optimized fractional order \(a^\ast\) leads to lower MSE values than the GFT case with \(a=1\). This verifies that the GFRFT domain can provide a more suitable spectral representation for graph signal reconstruction in several noiseless cases. Across the two noiseless signal settings, the optimized fractional order generally reduces the reconstruction error compared with the fixed \(a=1\) case. The amount of improvement varies with the signal type, spectral response, and reconstruction prior. In particular, the results show that the benefit of the GFRFT domain is more pronounced when the GFT domain is less matched to the signal structure, while some methods obtain \(a^\ast=1\), indicating that the GFT domain is already the optimal choice for those specific settings.

Fig.~\ref{fig:noiseless_reconstruction_visualization} visualizes the representative noiseless reconstruction results. For each simulated signal, the reconstructed signal maps are first compared with the corresponding input signal to show the overall recovery quality. The node wise reconstruction error maps further display the absolute reconstruction error at each graph node, thereby indicating where the reconstruction deviates more strongly from the original signal. In addition, the error improvement map shows the difference between the GFT domain reconstruction error and the optimized GFRFT domain reconstruction error. Positive values indicate that the optimized fractional domain yields a smaller local reconstruction error, while negative values indicate that the GFT domain performs better at the corresponding node.

\begin{figure*}[!t]
	\centering
	\setlength{\tabcolsep}{1.2pt}
	\renewcommand{\arraystretch}{0.95}
	
	\begin{tabular}{ccccc}
		\includegraphics[width=0.200\textwidth]{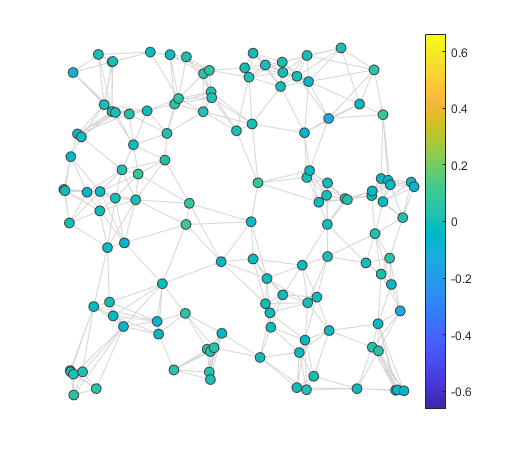} &
		\includegraphics[width=0.200\textwidth]{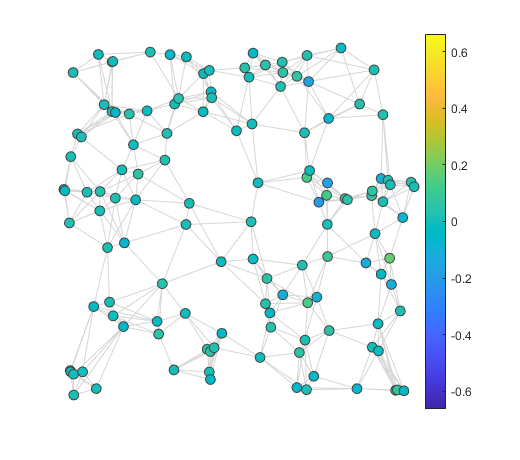} &
		\includegraphics[width=0.200\textwidth]{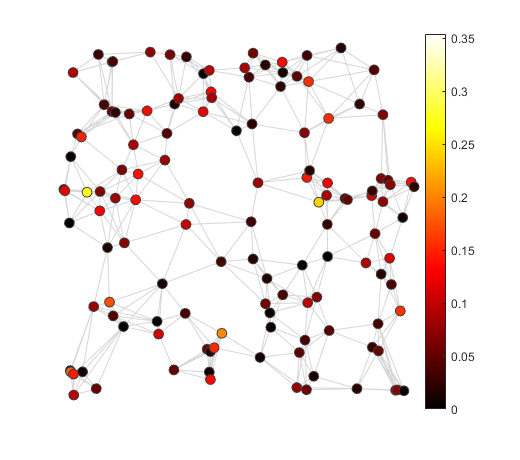} &
		\includegraphics[width=0.200\textwidth]{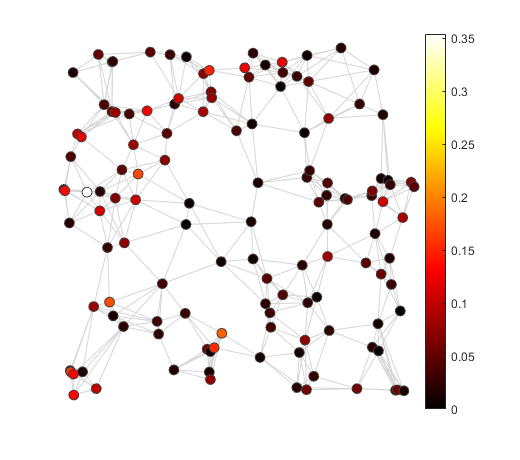} &
		\includegraphics[width=0.200\textwidth]{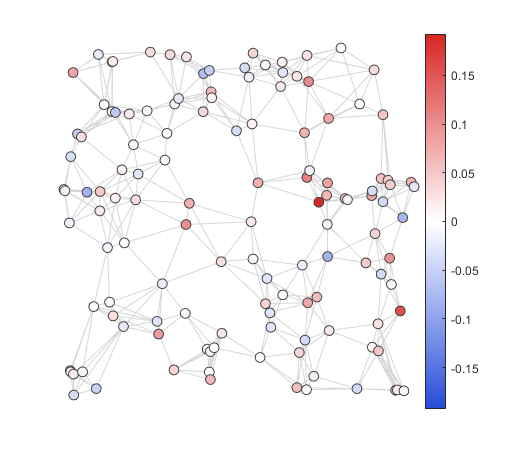} \\
		\footnotesize (a) NS input &
		\footnotesize (b) NS rec. &
		\footnotesize (c) NS GFT err. &
		\footnotesize (d) NS GFRFT err. &
		\footnotesize (e) NS improve \\
		
		\includegraphics[width=0.200\textwidth]{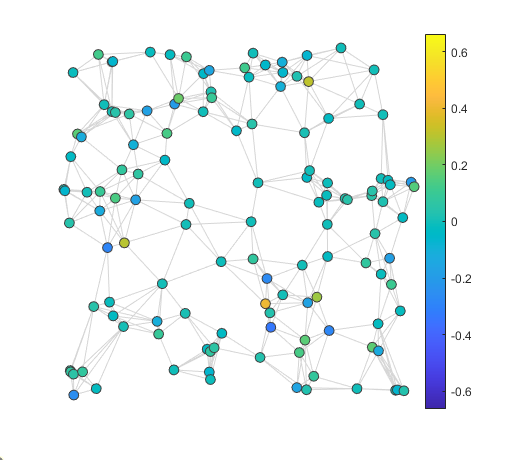} &
		\includegraphics[width=0.200\textwidth]{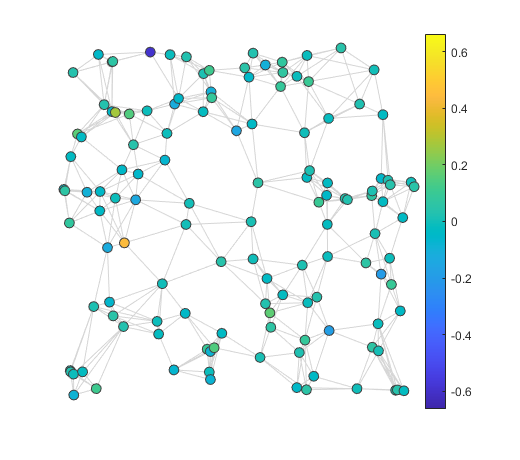} &
		\includegraphics[width=0.200\textwidth]{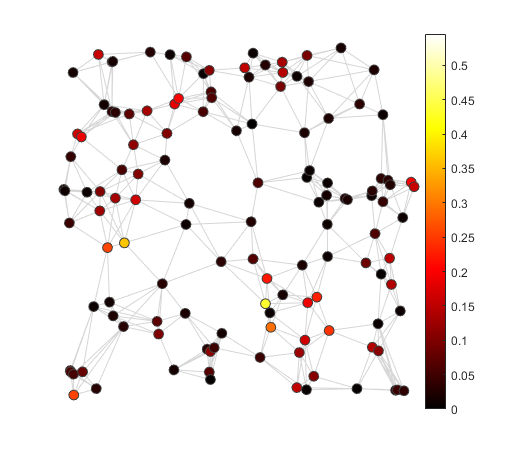} &
		\includegraphics[width=0.200\textwidth]{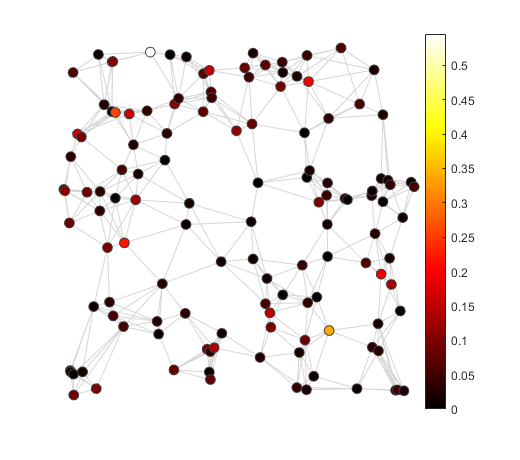} &
		\includegraphics[width=0.200\textwidth]{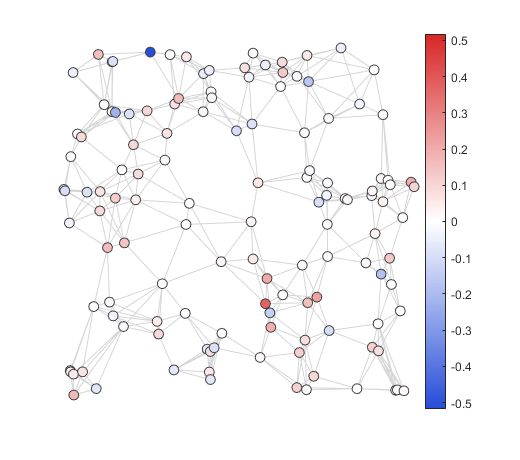} \\
		\footnotesize (f) SL input &
		\footnotesize (g) SL rec. &
		\footnotesize (h) SL GFT err. &
		\footnotesize (i) SL GFRFT err. &
		\footnotesize (j) SL improve
	\end{tabular}
	
	\caption{Representative visualization results for the noiseless simulated experiments. The first row shows the near-subspace case under the \(G_1\) response with SM UNC reconstruction, and the second row shows the smooth-local-anomaly case under the \(G_2\) response with ST UNC reconstruction.}
	\label{fig:noiseless_reconstruction_visualization}
\end{figure*}

\begin{table*}[!t]
	\centering
	\caption{Reconstruction Errors and Improvements for Noiseless Simulated Signals under \(G_1\) and \(G_2\).}
	\label{tab:noiseless_simulated_summary}
	\scriptsize
	\renewcommand{\arraystretch}{1.20}
	\setlength{\tabcolsep}{0pt}
	\begin{tabular*}{\textwidth}{@{\extracolsep{\fill}}lcccccccc@{}}
		\toprule
		\textbf{Method}
		& \multicolumn{4}{c}{\textbf{\(G_1\)}} 
		& \multicolumn{4}{c}{\textbf{\(G_2\)}} \\
		\cmidrule(lr){2-5} \cmidrule(lr){6-9}
		& \(a=1\) MSE (dB)
		& \(a^\ast\) MSE (dB)
		& \(a^\ast\)
		& Imp. (dB)
		& \(a=1\) MSE (dB)
		& \(a^\ast\) MSE (dB)
		& \(a^\ast\)
		& Imp. (dB) \\
		\midrule
		
		\multicolumn{9}{@{}l}{\textit{Near-subspace signal}} \\
		SS UNC      & \(-22.0498\) & \(-23.0740\) & \(0.26\) & \(1.0242\) & \(-22.5292\) & \(-23.3880\) & \(0.36\) & \(0.8588\) \\
		SS PD & \(-23.0729\) & \(-23.5454\) & \(0.54\) & \(0.4725\) & \(-23.3380\) & \(-23.8352\) & \(0.54\) & \(0.4972\) \\
		SM UNC      & \(-21.4784\) & \(-22.8736\) & \(0.14\) & \(1.3952\) & \(-21.6883\) & \(-22.9269\) & \(0.22\) & \(1.2385\) \\
		SM PD    & \(-23.0298\) & \(-23.5104\) & \(0.54\) & \(0.4805\) & \(-23.3141\) & \(-23.8172\) & \(0.54\) & \(0.5031\) \\
		ST UNC      & \(-27.7226\) & \(-27.7226\) & \(1.00\) & \(0.0000\) & \(-28.1012\) & \(-28.1012\) & \(1.00\) & \(0.0000\) \\
		ST PD  & \(-23.2845\) & \(-23.5502\) & \(0.66\) & \(0.2657\) & \(-23.4329\) & \(-23.8314\) & \(0.58\) & \(0.3985\) \\
		
		\addlinespace[2pt]
		\midrule
		
		\multicolumn{9}{@{}l}{\textit{Smooth-local-anomaly signal}} \\
		SS UNC      & \(-27.6358\) & \(-27.7555\) & \(0.94\) & \(0.1198\) & \(-26.1072\) & \(-26.2997\) & \(0.90\) & \(0.1925\) \\
		SS PD & \(-24.8767\) & \(-25.1250\) & \(0.84\) & \(0.2483\) & \(-24.9727\) & \(-25.2319\) & \(0.84\) & \(0.2592\) \\
		SM UNC      & \(-28.6939\) & \(-28.7598\) & \(0.96\) & \(0.0659\) & \(-27.5025\) & \(-27.6286\) & \(0.94\) & \(0.1261\) \\
		SM PD    & \(-24.9163\) & \(-25.1557\) & \(0.86\) & \(0.2394\) & \(-24.9878\) & \(-25.2446\) & \(0.84\) & \(0.2568\) \\
		ST UNC      & \(-18.6240\) & \(-20.0565\) & \(0.34\) & \(1.4325\) & \(-19.3742\) & \(-20.7978\) & \(0.32\) & \(1.4236\) \\
		ST PD  & \(-24.1143\) & \(-24.3802\) & \(0.82\) & \(0.2659\) & \(-24.6642\) & \(-24.9402\) & \(0.82\) & \(0.2761\) \\
		\bottomrule
	\end{tabular*}
\end{table*}
\subsubsection{Results on Noisy Simulated Signals}

The noisy experiments are performed by adding Gaussian white noise to the original graph signals. The remaining settings are the same as those in the noiseless experiments, including two signal types, two spectral response functions, and six reconstruction methods.

Table~\ref{tab:noisy_simulated_summary} summarizes the reconstruction results under noisy simulated settings. Compared with the noiseless case, the added Gaussian noise increases the reconstruction difficulty, but the optimized fractional order \(a^\ast\) still reduces the MSE for most reconstruction methods. This indicates that the GFRFT domain can provide a more adaptive spectral representation even when the observed graph signals are corrupted by noise. The improvement varies across signal types, spectral responses, and reconstruction priors.

Fig.~\ref{fig:noisy_reconstruction_visualization} shows the representative noisy reconstruction results. For each case, the visualization presents the noisy input or reconstructed signal, the node wise reconstruction errors, and the corresponding error improvement map.

\begin{figure*}[!t]
	\centering
	\setlength{\tabcolsep}{0pt}
	\renewcommand{\arraystretch}{0.95}
	
	\begin{tabular}{@{}ccccc@{}}
		\includegraphics[width=0.200\textwidth]{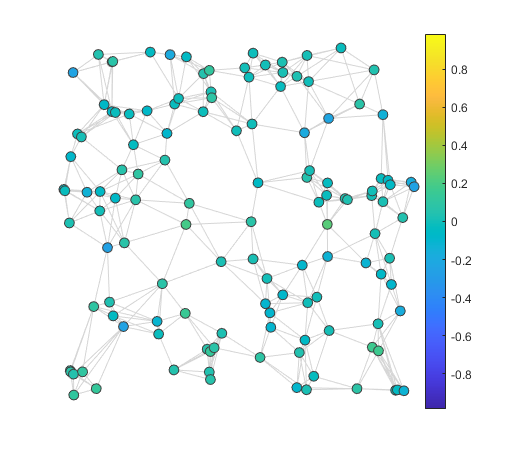} &
		\includegraphics[width=0.200\textwidth]{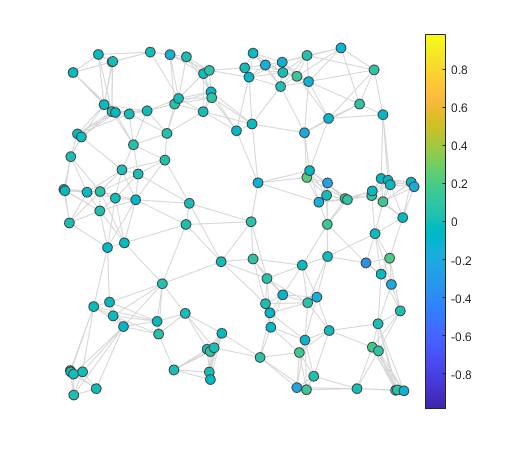} &
		\includegraphics[width=0.200\textwidth]{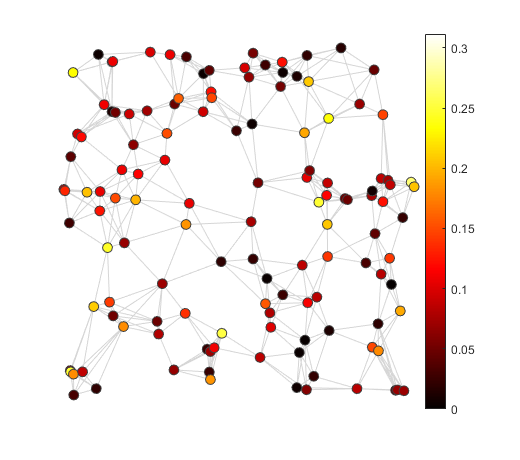} &
		\includegraphics[width=0.200\textwidth]{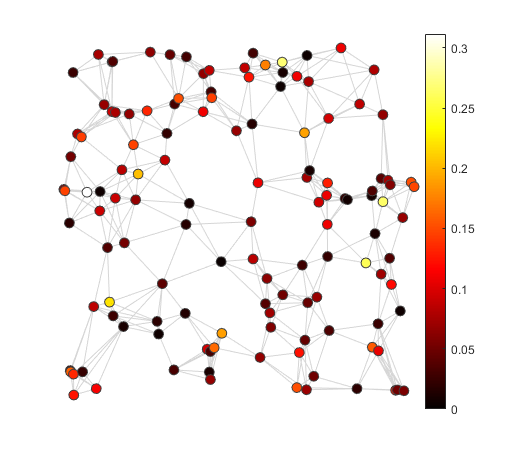} &
		\includegraphics[width=0.200\textwidth]{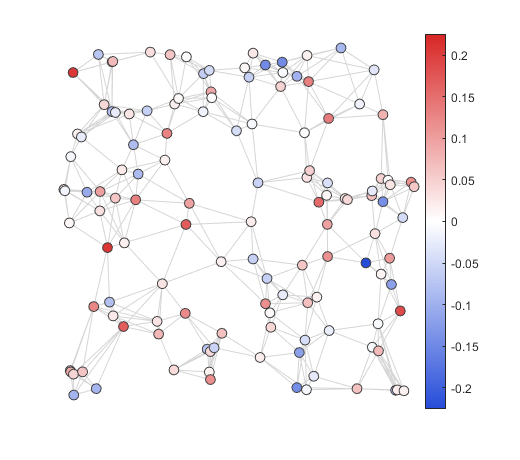} \\
		\footnotesize (a) NS input &
		\footnotesize (b) NS rec. &
		\footnotesize (c) NS GFT err. &
		\footnotesize (d) NS GFRFT err. &
		\footnotesize (e) NS improve \\
		
		\includegraphics[width=0.200\textwidth]{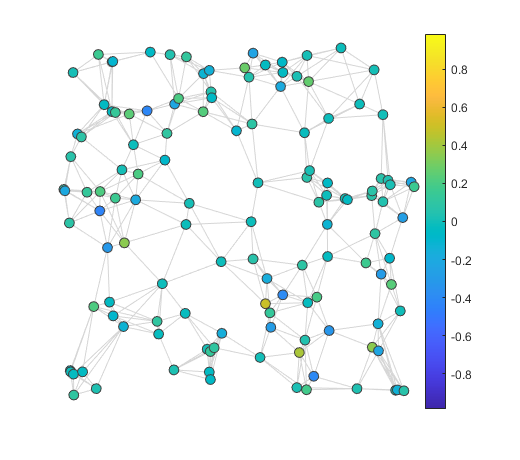} &
		\includegraphics[width=0.200\textwidth]{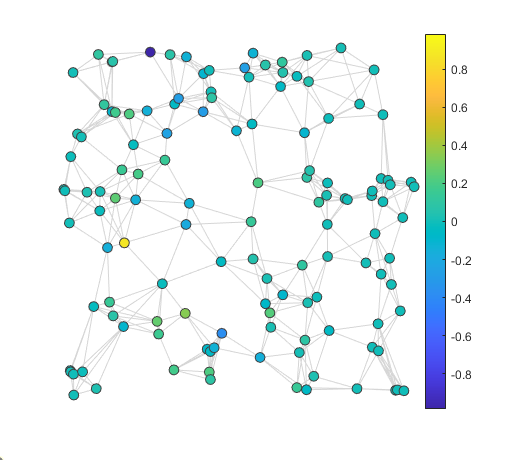} &
		\includegraphics[width=0.200\textwidth]{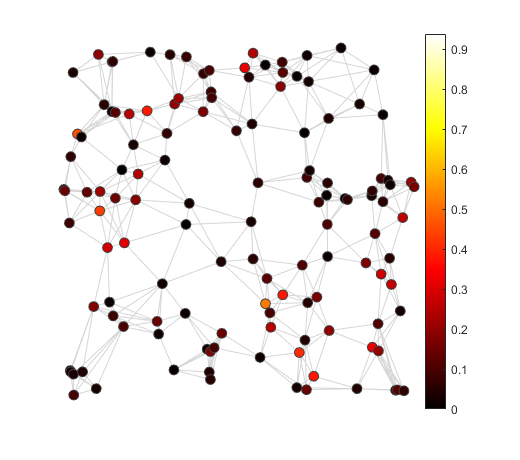} &
		\includegraphics[width=0.200\textwidth]{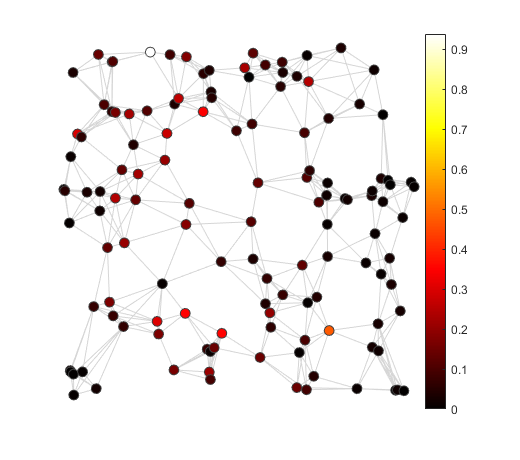} &
		\includegraphics[width=0.200\textwidth]{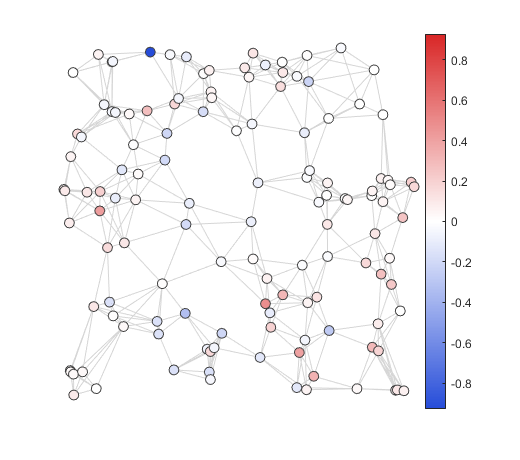} \\
		\footnotesize (f) SL input &
		\footnotesize (g) SL rec. &
		\footnotesize (h) SL GFT err. &
		\footnotesize (i) SL GFRFT err. &
		\footnotesize (j) SL improve
	\end{tabular}
	
	\caption{Representative visualization results for the noisy simulated experiments. The first row shows the near-subspace case under the \(G_1\) response with SM UNC reconstruction, and the second row shows the smooth-local-anomaly case under the \(G_2\) response with ST UNC reconstruction.}
	\label{fig:noisy_reconstruction_visualization}
\end{figure*}

\begin{table*}[!t]
	\centering
	\caption{Reconstruction Errors and Improvements for Noisy Simulated Signals under \(G_1\) and \(G_2\).}
	\label{tab:noisy_simulated_summary}
	\scriptsize
	\renewcommand{\arraystretch}{1.20}
	\setlength{\tabcolsep}{0pt}
	\begin{tabular*}{\textwidth}{@{\extracolsep{\fill}}lcccccccc@{}}
		\toprule
		\textbf{Method}
		& \multicolumn{4}{c}{\textbf{\(G_1\)}} 
		& \multicolumn{4}{c}{\textbf{\(G_2\)}} \\
		\cmidrule(lr){2-5} \cmidrule(lr){6-9}
		& \(a=1\) MSE (dB)
		& \(a^\ast\) MSE (dB)
		& \(a^\ast\)
		& Imp. (dB)
		& \(a=1\) MSE (dB)
		& \(a^\ast\) MSE (dB)
		& \(a^\ast\)
		& Imp. (dB) \\
		\midrule
		
		\multicolumn{9}{@{}l}{\textit{Near-subspace signal}} \\
		SS UNC      & \(-19.2872\) & \(-20.1525\) & \(0.14\) & \(0.8652\) & \(-19.4919\) & \(-20.2138\) & \(0.22\) & \(0.7219\) \\
		SS PD & \(-20.1280\) & \(-20.3723\) & \(0.32\) & \(0.2443\) & \(-20.0311\) & \(-20.3695\) & \(0.36\) & \(0.3384\) \\
		SM UNC      & \(-18.8446\) & \(-20.0471\) & \(0.10\) & \(1.2025\) & \(-18.9027\) & \(-19.9756\) & \(0.16\) & \(1.0728\) \\
		SM PD    & \(-20.2658\) & \(-20.5108\) & \(0.32\) & \(0.2450\) & \(-20.1330\) & \(-20.4705\) & \(0.36\) & \(0.3374\) \\
		ST UNC      & \(-18.7363\) & \(-18.7363\) & \(1.00\) & \(0.0000\) & \(-19.1302\) & \(-19.1302\) & \(1.00\) & \(0.0000\) \\
		ST PD  & \(-19.4025\) & \(-19.4754\) & \(0.32\) & \(0.0728\) & \(-19.6106\) & \(-19.8613\) & \(0.34\) & \(0.2507\) \\
		
		\addlinespace[2pt]
		\midrule
		
		\multicolumn{9}{@{}l}{\textit{Smooth--local-anomaly signal}} \\
		SS UNC      & \(-21.0276\) & \(-21.1642\) & \(0.86\) & \(0.1366\) & \(-20.5054\) & \(-20.7397\) & \(0.76\) & \(0.2343\) \\
		SS PD & \(-20.3948\) & \(-20.6860\) & \(0.66\) & \(0.2913\) & \(-20.1909\) & \(-20.5202\) & \(0.62\) & \(0.3293\) \\
		SM UNC      & \(-21.1963\) & \(-21.3080\) & \(0.88\) & \(0.1116\) & \(-20.8021\) & \(-20.9696\) & \(0.84\) & \(0.1675\) \\
		SM PD    & \(-20.6113\) & \(-20.8917\) & \(0.68\) & \(0.2804\) & \(-20.3317\) & \(-20.6609\) & \(0.64\) & \(0.3292\) \\
		ST UNC      & \(-15.4950\) & \(-15.7830\) & \(0.00\) & \(0.2880\) & \(-16.1010\) & \(-16.4200\) & \(0.00\) & \(0.3190\) \\
		ST PD  & \(-19.0792\) & \(-19.2385\) & \(0.76\) & \(0.1593\) & \(-19.5581\) & \(-19.7727\) & \(0.72\) & \(0.2146\) \\
		\bottomrule
	\end{tabular*}
\end{table*}

\subsection{Real Data Experiments}

\subsubsection{Data Source and Experimental Setting}

To further evaluate the applicability of the proposed method to real graph signals, the PeMS-BAY traffic sensor dataset~\cite{43} is used in the real data experiment. The dataset is constructed from a traffic sensor network, where each sensor corresponds to a graph node and the road connectivity between sensors is used to define the graph adjacency matrix. Different from the simulated experiments, the graph signal is not synthetically generated. Instead, a real traffic observation at a specific time step is selected as the target graph signal to be reconstructed.

To ensure that the sampling number \(K\) is compatible with the graph size, \(N=320\) nodes are selected from the original sensor set in the main real data experiment. The target graph signal is extracted from the first feature column of the PeMS-BAY observation matrix at the 10000-th time step. The spectral response functions \(G_1\) and \(G_2\) are kept the same as those used in the simulated experiments. A fixed random seed is used to ensure reproducibility.

\subsubsection{Reconstruction Methods on Real Data}

In the main real-data experiment, the sampling number is fixed as $K=64$,
which corresponds to the sampling ratio
\begin{equation*}
	\frac{K}{N}=\frac{64}{320}=0.2.
\end{equation*}
Thus, only 20\% of the graph samples are used to reconstruct the complete traffic sensor graph signal.

In addition to the SS-, SM-, and ST-based reconstruction methods, three baseline methods are included for comparison, namely BL-GFT~\cite{25}, Laplacian-Tikhonov~\cite{41} and Graph KRR~\cite{42}. These baselines are used to evaluate the fractional domain reconstruction framework against representative graph signal reconstruction schemes.

\subsubsection{Real Data Experimental Results}

Table~\ref{tab:real_data_k64_summary} reports the reconstruction results on the PeMS-BAY dataset with \(K=64\). The optimized fractional order generally yields lower MSE values than the fixed \(a=1\) setting, indicating that the GFRFT domain can also improve reconstruction performance on real traffic sensor graph signals. The improvement is method dependent: some methods show relatively small gains, while the stochastic prior based unconstrained method achieves a more evident reduction under both \(G_1\) and \(G_2\). Compared with the baseline methods, the GFRFT domain reconstruction methods obtain lower reconstruction errors in most cases, which supports the applicability of the proposed framework to real graph signal reconstruction.

\begin{table*}[!t]
	\centering
	\caption{Reconstruction Errors and Improvements for Real PeMS-BAY Data with \(K=64\).}
	\label{tab:real_data_k64_summary}
	\scriptsize
	\renewcommand{\arraystretch}{1.12}
	\setlength{\tabcolsep}{0pt}
	\begin{tabular*}{\textwidth}{@{\extracolsep{\fill}}lcccccccc@{}}
		\toprule
		\textbf{Method}
		& \multicolumn{4}{c}{\textbf{\(G_1\)}} 
		& \multicolumn{4}{c}{\textbf{\(G_2\)}} \\
		\cmidrule(lr){2-5} \cmidrule(lr){6-9}
		& \(a=1\) MSE (dB)
		& Best MSE (dB)
		& \(a^\ast\)
		& Imp. (dB)
		& \(a=1\) MSE (dB)
		& Best MSE (dB)
		& \(a^\ast\)
		& Imp. (dB) \\
		\midrule
		
		SS UNC & \(-25.9884\) & \(-26.0036\) & \(0.88\) & \(0.0152\)
		& \(-25.9195\) & \(-25.9228\) & \(0.92\) & \(0.0033\) \\
		SS PD  & \(-25.7804\) & \(-25.7839\) & \(0.90\) & \(0.0035\)
		& \(-25.8414\) & \(-25.8424\) & \(0.94\) & \(0.0010\) \\
		SM UNC & \(-26.0552\) & \(-26.0681\) & \(0.90\) & \(0.0129\)
		& \(-25.9784\) & \(-25.9865\) & \(0.90\) & \(0.0081\) \\
		SM PD  & \(-25.7836\) & \(-25.7874\) & \(0.88\) & \(0.0038\)
		& \(-25.8421\) & \(-25.8432\) & \(0.94\) & \(0.0011\) \\
		ST UNC & \(-25.9626\) & \(-26.8479\) & \(0.04\) & \(0.8853\)
		& \(-25.8975\) & \(-26.7713\) & \(0.04\) & \(0.8738\) \\
		ST PD  & \(-25.7811\) & \(-25.7825\) & \(0.92\) & \(0.0014\)
		& \(-25.8410\) & \(-25.8418\) & \(0.94\) & \(0.0008\) \\
		
		\midrule
		
		BL-GFT
		& -- & \(-24.7837\) & -- & --
		& -- & \(-24.7837\) & -- & -- \\
		Laplacian-Tikhonov
		& -- & \(-25.2105\) & -- & --
		& -- & \(-25.2105\) & -- & -- \\
		Graph KRR
		& -- & \(-25.5214\) & -- & --
		& -- & \(-25.5214\) & -- & -- \\
		
		\bottomrule
	\end{tabular*}
\end{table*}
\subsection{Sampling Rate Sensitivity Analysis}
\label{subsec:sampling_rate_sensitivity}

To analyze the influence of the sampling number on real graph signal reconstruction, we conduct a sampling number sensitivity analysis on the PeMS-BAY dataset by varying \(K\). Since the graph size is fixed as \(N=320\), three sampling numbers are considered:
\begin{equation*}
	K=32,\quad K=64,\quad K=160,
\end{equation*}
which correspond to sampling ratios of 0.1, 0.2, and 0.5, respectively. The case \(K=64\) is used as the main real data setting, while \(K=32\) and \(K=160\) are used to evaluate the reconstruction behavior under lower and higher sampling rates.

 All settings except the sampling number are kept unchanged, including the graph structure, the selected traffic signal, the spectral response functions, the fractional order search range, and the error metric. For each value of \(K\), the reconstruction error obtained with \(a=1\) is compared with the minimum error obtained at the optimized fractional order \(a^\ast\).

Table~\ref{tab:k_ablation_summary} summarizes the optimal reconstruction results for different sampling numbers. When \(K=32\), the available sampling information is limited, and the performance differences among different methods are relatively small. When \(K=160\), the reconstruction errors are substantially reduced, indicating that more sampled observations improve the recovery quality of real graph signals. Moreover, the optimal fractional order is not always equal to 1, which suggests that the GFT domain is not necessarily the optimal reconstruction domain under different sampling conditions. These results verify that the proposed GFRFT framework can adapt its spectral representation according to the sampling level.

Overall, the ablation results show that both the sampling number and the spectral response function influence the reconstruction performance. The improvement obtained by the GFRFT domain is therefore not determined by a fixed fractional order, but depends on the sampling condition, spectral response, and recovery model.
\begin{table*}[!t]
	\centering  
	\caption{Reconstruction Results on Real PeMS-BAY Data with Different Sampling Numbers \(K\)}
	\label{tab:k_ablation_summary}
	\scriptsize
	\renewcommand{\arraystretch}{1.12}
	\setlength{\tabcolsep}{0pt}
	\begin{tabular*}{\textwidth}{@{\extracolsep{\fill}}lcccccc@{}}
		\toprule
		\textbf{Method}
		& \multicolumn{3}{c}{\textbf{\(G_1\): \(a^\ast/\mathrm{MSE}_{\mathrm{dB}}\)}} 
		& \multicolumn{3}{c}{\textbf{\(G_2\): \(a^\ast/\mathrm{MSE}_{\mathrm{dB}}\)}} \\
		\cmidrule(lr){2-4} \cmidrule(lr){5-7}
		& \textbf{\(K=32\)}
		& \textbf{\(K=64\)}
		& \textbf{\(K=160\)}
		& \textbf{\(K=32\)}
		& \textbf{\(K=64\)}
		& \textbf{\(K=160\)} \\
		\midrule
		
		SS UNC 
		& \(0.84/{-25.6406}\)
		& \(0.92/{-25.9228}\)
		& \(0.46/{-28.0416}\)
		& \(0.84/{-25.6209}\)
		& \(0.92/{-25.9228}\)
		& \(0.50/{-28.1185}\) \\
		
		SS PD 
		& \(0.84/{-25.5467}\)
		& \(0.94/{-25.8424}\)
		& \(0.50/{-27.8502}\)
		& \(0.82/{-25.5684}\)
		& \(0.94/{-25.8424}\)
		& \(0.46/{-28.0629}\) \\
		
		SM UNC 
		& \(0.86/{-25.6774}\)
		& \(0.90/{-25.9865}\)
		& \(0.50/{-27.9962}\)
		& \(0.86/{-25.6582}\)
		& \(0.90/{-25.9865}\)
		& \(0.52/{-28.0262}\) \\
		
		SM PD 
		& \(0.84/{-25.5477}\)
		& \(0.94/{-25.8432}\)
		& \(0.50/{-27.8470}\)
		& \(0.82/{-25.5690}\)
		& \(0.94/{-25.8432}\)
		& \(0.46/{-28.0627}\) \\
		
		ST UNC 
		& \(0.92/{-25.6223}\)
		& \(0.04/{-26.7713}\)
		& \(0.06/{-29.8955}\)
		& \(0.10/{-25.5936}\)
		& \(0.04/{-26.7713}\)
		& \(0.08/{-29.3739}\) \\
		
		ST PD 
		& \(0.84/{-25.5466}\)
		& \(0.94/{-25.8418}\)
		& \(0.48/{-27.8491}\)
		& \(0.82/{-25.5687}\)
		& \(0.94/{-25.8418}\)
		& \(0.44/{-28.0659}\) \\
		
		\bottomrule
	\end{tabular*}
\end{table*}

\section{Conclusion}

In this paper, a graph signal sampling and reconstruction framework based on the GFRFT domain was investigated. By comparing the GFT domain with the optimized GFRFT domain, the experimental results demonstrate that the introduction of an adjustable fractional order can provide a more adaptive spectral representation for graph signal reconstruction. The results show that the optimized fractional domain generally achieves lower reconstruction errors than the conventional GFT domain. In addition, the ablation study with different sampling numbers further indicates that the optimal fractional order is affected by the sampling condition, spectral response, and reconstruction model, rather than being fixed in advance. These findings verify the effectiveness and flexibility of the proposed GFRFT based reconstruction framework, and suggest that fractional spectral representations can serve as a promising tool for graph signal sampling and recovery tasks.

\appendices

\section{Proofs of Propositions~\ref{prop:subspace_based_reconstruction}--\ref{prop:stochastic_reconstruction}}
\label{app:proof_propositions}

The three propositions follow from least-squares reconstruction, projection onto a prescribed reconstruction space, and second order statistical linear estimation. Let
\[
\boldsymbol{\Pi}_a
=
\boldsymbol{\Phi}_a
\left(
\boldsymbol{\Phi}_a^{\mathrm{H}}
\boldsymbol{\Phi}_a
\right)^{-1}
\boldsymbol{\Phi}_a^{\mathrm{H}}
\]
denote the least-squares projection onto the column space of
\(\boldsymbol{\Phi}_a\). Once the unconstrained reconstruction is obtained,
the predefined reconstruction is obtained by applying
\(\boldsymbol{\Pi}_a\) to the corresponding unconstrained result.

For Proposition~\ref{prop:subspace_based_reconstruction}, the subspace prior gives
\(\mathbf{x}=\mathbf{M}\mathbf{d}\) and, in the noiseless case,
\(\mathbf{y}_a=\mathbf{S}_a^{\ast}\mathbf{M}\mathbf{d}\). Since
\(\mathbf{S}_a^{\ast}\mathbf{M}\) has full column rank,
\[
\left(
\mathbf{S}_a^{\ast}\mathbf{M}
\right)^{\dagger}
\left(
\mathbf{S}_a^{\ast}\mathbf{M}
\right)
=
\mathbf{I}.
\]
Thus,
\[
\widetilde{\mathbf{x}}_{a,\mathrm{SS}}^{\mathrm{UNC}}
=
\mathbf{M}
\left(
\mathbf{S}_a^{\ast}\mathbf{M}
\right)^{\dagger}
\mathbf{y}_a
=
\mathbf{M}\mathbf{d}
=
\mathbf{x},
\]
which gives \eqref{eq:subspace_reconstruction} and
\eqref{eq:subspace_unc_hw}. Applying \(\boldsymbol{\Pi}_a\) to
\eqref{eq:subspace_reconstruction} gives
\eqref{eq:subspace_pre_h} and \eqref{eq:subspace_pre_reconstruction}.

For Proposition~\ref{prop:smoothness_based_reconstruction}, the smoothness aaainduced reconstruction is written as
\(\widetilde{\mathbf{x}}_{a,\mathrm{SM}}^{\mathrm{UNC}}
=
\mathbf{R}_{a,\mathrm{SM}}\widehat{\mathbf{c}}\). The full column rank condition of
\(\mathbf{S}_a^{\ast}\mathbf{R}_{a,\mathrm{SM}}\) gives the least-squares solution
\[
\widehat{\mathbf{c}}
=
\left(
\mathbf{S}_a^{\ast}
\mathbf{R}_{a,\mathrm{SM}}
\right)^{\dagger}
\mathbf{y}_a .
\]
Substitution gives \eqref{eq:smoothness_unc_reconstruction} and
\eqref{eq:smoothness_unc_hw}. Applying \(\boldsymbol{\Pi}_a\) gives
\eqref{eq:smoothness_pre_h} and \eqref{eq:smoothness_pre_reconstruction}.

For Proposition~\ref{prop:stochastic_reconstruction}, the unconstrained stochastic linear estimator can be written from the second-order statistics as
\[
\widetilde{\mathbf{x}}_{a,\mathrm{ST}}^{\mathrm{UNC}}
=
\mathbb{E}
\left[
\mathbf{x}\mathbf{y}_a^{\mathrm{H}}
\right]
\left(
\mathbb{E}
\left[
\mathbf{y}_a\mathbf{y}_a^{\mathrm{H}}
\right]
\right)^{-1}
\mathbf{y}_a .
\]
Using \eqref{eq:gfrft_sampled_signal} and the independence between
\(\mathbf{x}\) and \(\boldsymbol{\eta}\), we have
\[
\mathbb{E}
\left[
\mathbf{x}\mathbf{y}_a^{\mathrm{H}}
\right]
=
\boldsymbol{\Gamma}_x
\left(
\mathbf{S}_a^{\ast}
\right)^{\mathrm{H}},
\qquad
\mathbb{E}
\left[
\mathbf{y}_a\mathbf{y}_a^{\mathrm{H}}
\right]
=
\mathbf{S}_a^{\ast}
\boldsymbol{\Gamma}_x
\left(
\mathbf{S}_a^{\ast}
\right)^{\mathrm{H}}
+
\boldsymbol{\Gamma}_{\eta}.
\]
Substituting these covariance identities into the stochastic linear estimator gives
\eqref{eq:stochastic_unc_reconstruction} and \eqref{eq:stochastic_unc_hw}.
Applying \(\boldsymbol{\Pi}_a\) gives \eqref{eq:stochastic_pre_h} and
\eqref{eq:stochastic_pre_reconstruction}. This completes the proof.\qed

\section{Proof of Theorem~\ref{thm:fractional_reconstruction_advantage}}
\label{app:proof_fractional_reconstruction_advantage}

Substituting \eqref{eq:fractional_sampling_observation} into
\eqref{eq:fractional_ls_reconstruction}, and using
\(\mathbf{A}_a^{\dagger}\mathbf{A}_a=\mathbf{I}\), we obtain
\begin{align}
	\widetilde{\mathbf{x}}_a
	&=
	\mathbf{U}_{K,a}
	\mathbf{A}_a^{\dagger}
	\left(
	\mathbf{A}_a\mathbf{c}_a
	+
	\mathbf{S}_a^{\ast}\mathbf{r}_{K,a}
	+
	\boldsymbol{\eta}
	\right) \nonumber\\
	&=
	\mathbf{x}_{K,a}
	+
	\mathbf{U}_{K,a}
	\mathbf{A}_a^{\dagger}
	\mathbf{S}_a^{\ast}
	\mathbf{r}_{K,a}
	+
	\mathbf{U}_{K,a}
	\mathbf{A}_a^{\dagger}
	\boldsymbol{\eta}.
	\label{eq:app_reconstruction_expansion}
\end{align}
Thus, from
\(\mathbf{x}=\mathbf{x}_{K,a}+\mathbf{r}_{K,a}\), the reconstruction error is
\begin{equation}
	\mathbf{x}
	-
	\widetilde{\mathbf{x}}_a
	=
	\mathbf{r}_{K,a}
	-
	\mathbf{U}_{K,a}
	\mathbf{A}_a^{\dagger}
	\mathbf{S}_a^{\ast}
	\mathbf{r}_{K,a}
	-
	\mathbf{U}_{K,a}
	\mathbf{A}_a^{\dagger}
	\boldsymbol{\eta}.
	\label{eq:app_error_decomposition}
\end{equation}

Since
\(\mathbf{r}_{K,a}\in\mathcal{U}_{K,a}^{\perp}\), while
\(\mathbf{U}_{K,a}\mathbf{A}_a^{\dagger}\mathbf{S}_a^{\ast}\mathbf{r}_{K,a}\)
and
\(\mathbf{U}_{K,a}\mathbf{A}_a^{\dagger}\boldsymbol{\eta}\)
belong to \(\mathcal{U}_{K,a}\), the residual term is orthogonal to the two
low-dimensional error terms. Hence,
\begin{equation}
	\left\|
	\mathbf{x}
	-
	\widetilde{\mathbf{x}}_a
	\right\|_2^2
	=
	\left\|
	\mathbf{r}_{K,a}
	\right\|_2^2
	+
	\left\|
	\mathbf{A}_a^{\dagger}
	\mathbf{S}_a^{\ast}
	\mathbf{r}_{K,a}
	+
	\mathbf{A}_a^{\dagger}
	\boldsymbol{\eta}
	\right\|_2^2 .
	\label{eq:app_error_norm_decomposition}
\end{equation}
Taking expectation and using
\(\mathbb{E}[\boldsymbol{\eta}]=\mathbf{0}\) and
\(\mathbb{E}[\boldsymbol{\eta}\boldsymbol{\eta}^{\mathrm{H}}]
=
\sigma^2\mathbf{I}\), we have
\begin{align}
	\mathbb{E}
	\left[
	\left\|
	\mathbf{x}
	-
	\widetilde{\mathbf{x}}_a
	\right\|_2^2
	\right]
	&=
	\left\|
	\mathbf{r}_{K,a}
	\right\|_2^2
	+
	\left\|
	\mathbf{A}_a^{\dagger}
	\mathbf{S}_a^{\ast}
	\mathbf{r}_{K,a}
	\right\|_2^2
	+
	\sigma^2
	\left\|
	\mathbf{A}_a^{\dagger}
	\right\|_F^2 .
	\label{eq:app_expected_error}
\end{align}

Moreover, since
\(\mathbf{r}_{K,a}
=
(\mathbf{I}-\mathbf{P}_{K,a})\mathbf{r}_{K,a}\), the residual leakage term satisfies
\begin{equation}
	\left\|
	\mathbf{A}_a^{\dagger}
	\mathbf{S}_a^{\ast}
	\mathbf{r}_{K,a}
	\right\|_2
	\leq
	\left\|
	\mathbf{A}_a^{\dagger}
	\mathbf{S}_a^{\ast}
	(\mathbf{I}-\mathbf{P}_{K,a})
	\right\|_2
	\left\|
	\mathbf{r}_{K,a}
	\right\|_2
	=
	\tau_a
	\left\|
	\mathbf{r}_{K,a}
	\right\|_2 .
	\label{eq:app_leakage_bound}
\end{equation}
By the definition
\(\nu_a=\sigma^2\|\mathbf{A}_a^{\dagger}\|_F^2\), we further obtain
\begin{equation}
	\mathbb{E}
	\left[
	\left\|
	\mathbf{x}
	-
	\widetilde{\mathbf{x}}_a
	\right\|_2^2
	\right]
	\leq
	\left(
	1+\tau_a^2
	\right)
	\left\|
	\mathbf{r}_{K,a}
	\right\|_2^2
	+
	\nu_a .
	\label{eq:app_fractional_error_upper_bound}
\end{equation}

For the GFT domain reconstruction, since
\(\mathbf{P}_{K,1}\mathbf{x}\) is the best approximation of
\(\mathbf{x}\) in the \(K\)-dimensional GFT subspace, any reconstruction
\(\widetilde{\mathbf{x}}_1\) constrained to this subspace satisfies
\begin{equation}
	\mathbb{E}
	\left[
	\left\|
	\mathbf{x}
	-
	\widetilde{\mathbf{x}}_1
	\right\|_2^2
	\right]
	\geq
	\left\|
	\mathbf{x}
	-
	\mathbf{P}_{K,1}
	\mathbf{x}
	\right\|_2^2
	=
	\left\|
	\mathbf{r}_{K,1}
	\right\|_2^2 .
	\label{eq:app_gft_expected_lower_bound}
\end{equation}
Therefore, if the sufficient condition in
\eqref{eq:theorem_sufficient_condition} holds, then
\eqref{eq:app_fractional_error_upper_bound} and
\eqref{eq:app_gft_expected_lower_bound} imply
\[
\mathbb{E}
\left[
\left\|
\mathbf{x}
-
\widetilde{\mathbf{x}}_a
\right\|_2^2
\right]
<
\mathbb{E}
\left[
\left\|
\mathbf{x}
-
\widetilde{\mathbf{x}}_1
\right\|_2^2
\right],
\]
which is exactly the claim in
\eqref{eq:theorem_reconstruction_advantage}. This completes the proof.\qed
 
%

\bibliographystyle{IEEEtran}
\bibliography{references}

@article{1,
	title={Graph signal processing: Overview, challenges, and applications},
	author={Ortega, Antonio and Frossard, Pascal and Kova{\v{c}}evi{\'c}, Jelena and Moura, Jos{\'e} MF and Vandergheynst, Pierre},
	journal={Proceedings of the IEEE},
	volume={106},
	number={5},
	pages={808--828},
	year={2018},
	publisher={IEEE}
}

@article{2,
	title={Graph signal processing: History, development, impact, and outlook},
	author={Leus, Geert and Marques, Antonio G and Moura, Jos{\'e} MF and Ortega, Antonio and Shuman, David I},
	journal={IEEE Signal Processing Magazine},
	volume={40},
	number={4},
	pages={49--60},
	year={2023},
	publisher={IEEE}
}

@article{3,
	title={Discrete signal processing on graphs},
	author={Sandryhaila, Aliaksei and Moura, Jos{\'e} MF},
	journal={IEEE Transactions on Signal Processing},
	volume={61},
	number={7},
	pages={1644--1656},
	year={2013},
	publisher={IEEE}
}

@book{4,
	title={Introduction to Graph Signal Processing},
	author={Ortega, Antonio},
	year={2022},
	publisher={Cambridge University Press}
}

@article{5,
	title={Big data analysis with signal processing on graphs: Representation and processing of massive data sets with irregular structure},
	author={Sandryhaila, Aliaksei and Moura, Jos MF},
	journal={IEEE Signal Processing Magazine},
	volume={31},
	number={5},
	pages={80--90},
	year={2014},
	publisher={IEEE}
}

@article{6,
	title={The emerging field of signal processing on graphs: Extending high-dimensional data analysis to networks and other irregular domains},
	author={Shuman, David I and Narang, Sunil K and Frossard, Pascal and Ortega, Antonio and Vandergheynst, Pierre},
	journal={IEEE Signal Processing Magazine},
	volume={30},
	number={3},
	pages={83--98},
	year={2013},
	publisher={IEEE}
}

@article{7,
	title={Fast graph {Fourier} transforms based on graph symmetry and bipartition},
	author={Lu, Keng-Shih and Ortega, Antonio},
	journal={IEEE Transactions on Signal Processing},
	volume={67},
	number={18},
	pages={4855--4869},
	year={2019},
	publisher={IEEE}
}

@article{8,
	title={Graph {Fourier} transform: A stable approximation},
	author={Domingos, Joao and Moura, Jos{\'e} MF},
	journal={IEEE Transactions on Signal Processing},
	volume={68},
	pages={4422--4437},
	year={2020},
	publisher={IEEE}
}

@article{9,
	title={Generic reversible visible watermarking via regularized graph {Fourier} transform coding},
	author={Qi, Wenfa and Guo, Sirui and Hu, Wei},
	journal={IEEE Transactions on Image Processing},
	volume={31},
	pages={691--705},
	year={2021},
	publisher={IEEE}
}

@article{10,
	title={On the shift operator, graph frequency, and optimal filtering in graph signal processing},
	author={Gavili, Adnan and Zhang, Xiao-Ping},
	journal={IEEE Transactions on Signal Processing},
	volume={65},
	number={23},
	pages={6303--6318},
	year={2017},
	publisher={IEEE}
}

@article{11,
	title={Sampling, filtering and sparse approximations on combinatorial graphs},
	author={Pesenson, Isaac Z and Pesenson, Meyer Z},
	journal={Journal of Fourier Analysis and Applications},
	volume={16},
	number={6},
	pages={921--942},
	year={2010},
	publisher={Springer}
}

@article{12,
	title={Semi-supervised classification with graph convolutional networks},
	author={Kipf, Thomas N and Welling, Max},
	journal={arXiv preprint arXiv:1609.02907},
	year={2016}
}

@article{13,
	title={Inductive representation learning on large graphs},
	author={Hamilton, Will and Ying, Zhitao and Leskovec, Jure},
	journal={Advances in Neural Information Processing Systems},
	volume={30},
	year={2017}
}

@article{14,
	title={Sampling signals on graphs: From theory to applications},
	author={Tanaka, Yuichi and Eldar, Yonina C and Ortega, Antonio and Cheung, Gene},
	journal={IEEE Signal Processing Magazine},
	volume={37},
	number={6},
	pages={14--30},
	year={2020},
	publisher={IEEE}
}

@article{15,
	title={Discrete signal processing on graphs: Sampling theory},
	author={Chen, Siheng and Varma, Rohan and Sandryhaila, Aliaksei and Kova{\v{c}}evi{\'c}, Jelena},
	journal={IEEE Transactions on Signal Processing},
	volume={63},
	number={24},
	pages={6510--6523},
	year={2015},
	publisher={IEEE}
}

@article{16,
	title={Efficient sampling set selection for bandlimited graph signals using graph spectral proxies},
	author={Anis, Aamir and Gadde, Akshay and Ortega, Antonio},
	journal={IEEE Transactions on Signal Processing},
	volume={64},
	number={14},
	pages={3775--3789},
	year={2016},
	publisher={IEEE}
}

@article{17,
	title={Sampling of graph signals with successive local aggregations},
	author={Marques, Antonio G and Segarra, Santiago and Leus, Geert and Ribeiro, Alejandro},
	journal={IEEE Transactions on Signal Processing},
	volume={64},
	number={7},
	pages={1832--1843},
	year={2015},
	publisher={IEEE}
}

@article{18,
	title={Sampling in {Paley-Wiener} spaces on combinatorial graphs},
	author={Pesenson, Isaac},
	journal={Transactions of the American Mathematical Society},
	volume={360},
	number={10},
	pages={5603--5627},
	year={2008}
}

@article{19,
	title={Graph signal processing: Modulation, convolution, and sampling},
	author={Shi, John and Moura, Jose MF},
	journal={arXiv preprint arXiv:1912.06762},
	year={2019}
}

@inproceedings{22,
	title={Recovery of missing sensor data by reconstructing time-varying graph signals},
	author={Mondal, Anindya and Das, Mayukhmali and Chatterjee, Aditi and Venkateswaran, Palaniandavar},
	booktitle={2022 30th European Signal Processing Conference},
	pages={2181--2185},
	year={2022},
	organization={IEEE}
}

@article{23,
	title={Graph signal reconstruction techniques for iot air pollution monitoring platforms},
	author={Ferrer-Cid, Pau and Barcelo-Ordinas, Jose M and Garcia-Vidal, Jorge},
	journal={IEEE Internet of Things Journal},
	volume={9},
	number={24},
	pages={25350--25362},
	year={2022},
	publisher={IEEE}
}

@inproceedings{24,
	title={Piecewise stationary modeling of random processes over graphs with an application to traffic prediction},
	author={Hasanzadeh, Arman and Liu, Xi and Duffield, Nick and Narayanan, Krishna R},
	booktitle={2019 IEEE International Conference on Big Data},
    year= {2019},
    pages= {3779--3788}
}

@article{20,
	title={Spectral domain sampling of graph signals},
	author={Tanaka, Yuichi},
	journal={IEEE Transactions on Signal Processing},
	volume={66},
	number={14},
	pages={3752--3767},
	year={2018},
	publisher={IEEE}
}

@article{21,
	title={Two-channel critically sampled graph filter banks with spectral domain sampling},
	author={Sakiyama, Akie and Watanabe, Kana and Tanaka, Yuichi and Ortega, Antonio},
	journal={IEEE Transactions on Signal Processing},
	volume={67},
	number={6},
	pages={1447--1460},
	year={2019},
	publisher={IEEE}
}

@inproceedings{25,
	title={The fractional {Fourier} transform on graphs: Sampling and recovery},
	author={Wang, Yiqian and Li, Bingzhao},
	booktitle={2018 14th IEEE International Conference on Signal Processing},
	pages={1103--1108},
	year={2018},
}

@article{26,
	title={Optimal fractional {Fourier} filtering for graph signals},
	author={Ozturk, Cuneyd and Ozaktas, Haldun M and Gezici, Sinan and Ko{\c{c}}, Aykut},
	journal={IEEE Transactions on Signal Processing},
	volume={69},
	pages={2902--2912},
	year={2021},
	publisher={IEEE}
}

@article{27,
	title={Graph fractional {Fourier} transform: A unified theory},
	author={Alika{\c{s}}ifo{\u{g}}lu, Tuna and Kartal, B{\"u}nyamin and Ko{\c{c}}, Aykut},
	journal={IEEE Transactions on Signal Processing},
	volume={72},
	pages={3834--3850},
	year={2024},
	publisher={IEEE}
}

@article{28,
	title={Joint time-vertex fractional {Fourier} transform},
	author={Alika{\c{s}}ifo{\u{g}}lu, Tuna and Kartal, B{\"u}nyamin and {\"O}zg{\"u}nay, Eray and Ko{\c{c}}, Aykut},
	journal={Signal Processing},
	volume={233},
	pages={109944},
	year={2025},
	publisher={Elsevier}
}

@article{29,
	title={{FGFRFT}: Fast Graph Fractional {Fourier} Transform via Exact Spectral Splitting and {Fourier}-Series Approximation},
	author={Yan, Ziqi and Wang, Mingzhi and Shi, Sen and Zhao, Feiyue and Cui, Manjun and He, Yangfan and Zhang, Zhichao},
	journal={arXiv preprint arXiv:2602.20870},
	year={2026}
}

@article{291,
	title={Spectral graph fractional {Fourier} transform for directed graphs and its application},
	author={Yan, Fang-Jia and Li, Bing-Zhao},
	journal={Signal Processing},
	volume={210},
	pages={109099},
	year={2023},
	publisher={Elsevier}
}

@article{292,
	title={The graph fractional {Fourier} transform in {Hilbert} space},
	author={Zhang, Yu and Li, Bing-Zhao},
	journal={IEEE Transactions on Signal and Information Processing over Networks},
	year={2025},
	volume={11},
	number={},
	pages={242-257},
	publisher={IEEE}
}

@article{30,
	title={The Shannon sampling theorem—Its various extensions and applications: A tutorial review},
	author={Jerri, Abdul J},
	journal={Proceedings of the IEEE},
	volume={65},
	number={11},
	pages={1565--1596},
	year={1977},
	publisher={IEEE}
}

@book{31,
	title={Sampling Theory: Beyond Bandlimited Systems},
	author={Eldar, Yonina C},
	year={2015},
	publisher={Cambridge University Press}
}

@article{32,
	title={Beyond bandlimited sampling},
	author={Eldar, Yonina C and Michaeli, Tomer},
	journal={IEEE Signal Processing Magazine},
	volume={26},
	number={3},
	pages={48--68},
	year={2009},
	publisher={IEEE}
}

@article{33,
	author  = {Yamashita, Keitaro and Naganuma, Kazuki and Ono, Shunsuke},
	title   = {Generalized graph signal sampling by difference-of-convex optimization},
	journal = {arXiv preprint arXiv:2306.14634},
	year    = {2023}
}

@article{34,
	title={Generalized sampling on graphs with subspace and smoothness priors},
	author={Tanaka, Yuichi and Eldar, Yonina C},
	journal={IEEE Transactions on Signal Processing},
	volume={68},
	pages={2272--2286},
	year={2020},
	publisher={IEEE}
}

@article{35,
	title={Sampling Method for Generalized Graph Signals with Pre-selected Vertices via DC Optimization},
	author={Yamashita, Keitaro and Naganuma, Kazuki and Ono, Shunsuke},
	journal={IEEE Open Journal of Signal Processing},
	year={2026},
	volume={7},
	number={},
	pages={314-323},
	publisher={IEEE}
}

@inproceedings{36,
	title={Generalized graph signal sampling and reconstruction},
	author={Wang, Xiaohan and Chen, Jiaxuan and Gu, Yuantao},
	booktitle={2015 IEEE Global Conference on Signal and Information Processing},
	pages={567--571},
	year={2015},
}

@article{37,
	title={Nonideal sampling and interpolation from noisy observations in shift-invariant spaces},
	author={Eldar, Yonina C and Unser, Michael},
	journal={IEEE Transactions on Signal Processing},
	volume={54},
	number={7},
	pages={2636--2651},
	year={2006},
	publisher={IEEE}
}

@article{38,
	title={A minimum squared-error framework for generalized sampling},
	author={Eldar, Yonina C and Dvorkind, Tsvi G},
	journal={IEEE Transactions on Signal Processing},
	volume={54},
	number={6},
	pages={2155--2167},
	year={2006},
	publisher={IEEE}
}

@article{39,
	title={Graph signal sampling under stochastic priors},
	author={Hara, Junya and Tanaka, Yuichi and Eldar, Yonina C},
	journal={IEEE Transactions on Signal Processing},
	volume={71},
	pages={1421--1434},
	year={2023},
	publisher={IEEE}
}

@article{40,
	title={Generalized sampling of graph signals with the prior information based on graph fractional {Fourier} transform},
	author={Wei, Deyun and Yan, Zhenyang},
	journal={Signal Processing},
	volume={214},
	pages={109263},
	year={2024},
	publisher={Elsevier}
}

@article{41,
	title={Signal recovery on graphs: Variation minimization},
	author={Chen, Siheng and Sandryhaila, Aliaksei and Moura, Jos{\'e} MF and Kova{\v{c}}evi{\'c}, Jelena},
	journal={IEEE Transactions on Signal Processing},
	volume={63},
	number={17},
	pages={4609--4624},
	year={2015},
	publisher={IEEE}
}

@article{42,
	title={Kernel based reconstruction for generalized graph signal processing},
	author={Jian, Xingchao and Tay, Wee Peng and Eldar, Yonina C},
	journal={IEEE Transactions on Signal Processing},
	volume={72},
	pages={2308--2322},
	year={2024},
	publisher={IEEE}
}

@article{43,
	title={Diffusion convolutional recurrent neural network: Data-driven traffic forecasting},
	author={Li, Yaguang and Yu, Rose and Shahabi, Cyrus and Liu, Yan},
	journal={arXiv preprint arXiv:1707.01926},
	year={2017}
}

\end{document}